%% file: revised_manuscript.tex
\documentclass[aip,cha,reprint,amsmath,amssymb]{revtex4-2}
\usepackage{graphicx}
\usepackage{xcolor}
\colorlet{blue}{black}
\usepackage{dcolumn}
\usepackage{bm}
\usepackage[utf8]{inputenc}
\usepackage[T1]{fontenc}

\usepackage{etoolbox}
\usepackage{mathrsfs}
\usepackage{amsthm}
\usepackage{amsmath}
\usepackage{mathtools}
\usepackage{acro}
\acsetup{first-style=long-short}
\input{acronyms}
\usepackage{placeins}
\usepackage{algpseudocode}
\usepackage{newfloat}
\DeclareFloatingEnvironment[
    fileext=loa,
    listname={List of Algorithms},
    name=Algorithm,
    placement=tbhp,
]{algorithm}

\usepackage{tikz}
\usepackage{pgfplots}
\pgfplotsset{compat=1.18} 
\usetikzlibrary{arrows.meta}
\usetikzlibrary{decorations.pathmorphing,patterns}
\usepgfplotslibrary{groupplots}
\usepackage{newtxtext,newtxmath}
\counterwithout{figure}{section}

\usepackage[hidelinks]{hyperref}

\newtheorem{theorem}{Theorem}[section]
\newtheorem{lemma}[theorem]{Lemma}
\newtheorem{proposition}[theorem]{Proposition}

\theoremstyle{definition}
\newtheorem{definition}[theorem]{Definition}
\newtheorem{assumption}[theorem]{Assumption}
\theoremstyle{remark}

\begin{document}

\title{Geometric Mode-Selection Scores for Delay-Coordinates Dynamic Mode Decomposition}

\author{Yoav Harris}
\email{yoav.h@campus.technion.ac.il}
\affiliation{Viterbi Faculty of Electrical and Computer Engineering, Technion -- Israel Institute of Technology, Haifa 3200003, Israel}

\author{Hadas Benisty}
\affiliation{Rappaport Faculty of Medicine, Technion -- Israel Institute of Technology, Haifa 3200003, Israel}

\author{Ronen Talmon}
\affiliation{Viterbi Faculty of Electrical and Computer Engineering, Technion -- Israel Institute of Technology, Haifa 3200003, Israel}

\begin{abstract}
    Delay-coordinates dynamic mode decomposition (DC--DMD) is widely used to extract coherent spatiotemporal modes from high-dimensional time series. A central challenge is distinguishing dynamically meaningful modes from spurious modes induced by noise and order overestimation. We frame this as a \emph{mode-selection scoring} problem: each mode receives a score that ranks it as true or spurious; any hard selection (threshold or clustering) is a downstream choice.
    We show that mode selection in DC--DMD is fundamentally a problem of subspace geometry. True modes are characterized by concentration within a low-dimensional signal subspace, whereas spurious modes \textcolor{blue}{tend to retain} non-negligible components outside any moderate overestimate of that subspace. This geometric distinction defines true and spurious modes and motivates fully data-driven robust scoring criteria.
    The framework yields two complementary scores. The first uses a data-driven proxy of the signal subspace to compute a residual. The second comes from a new operator-theoretic analysis of delay embedding: using a block-companion formulation, we show that all modes exhibit a Kronecker--Vandermonde (KV) structure, with true modes distinguished by the degree of conformity to it. This deviation is \textcolor{blue}{governed by} the geometric residual.
    Our analysis further explains the empirical behavior of magnitude- and norm-based heuristics and clarifies when and why they fail under delay coordinates.
    Numerical experiments, evaluated by precision--recall AUC of true-vs-spurious ranking, show that the proposed scores \textcolor{blue}{outperform the tested baselines across most of the small-spatial-dimension regime}.
\end{abstract}

\maketitle

\begin{quotation}
Delay-coordinates dynamic mode decomposition (DC--DMD) is a widely-used tool for extracting coherent patterns from noisy, high-dimensional time series, but in practice it often returns more candidate components than truly reflect the underlying dynamics. The central challenge is to determine which recovered modes are physically meaningful and which are artifacts of noise and order overestimation. We address this as a mode-selection scoring problem: each computed mode receives a per-mode score that ranks it as true or spurious, with the final hard selection (threshold or clustering) treated as a downstream choice. We show that this scoring problem can be understood through subspace geometry, leading to a data-driven residual score that distinguishes true from spurious modes. We further show that delay coordinates impose a characteristic internal organization on the recovered modes, yielding a complementary structure-based score. This geometric framework also clarifies why commonly used magnitude- and norm-based heuristics can become unreliable. Together, these results provide a more interpretable and robust set of mode-selection scores for DC--DMD.
\end{quotation}
    
\input{sections/introduction}
\input{sections/problem_setting}
\input{sections/related_work}
\input{sections/mode_geometry}

\input{sections/kv_structure_and_nested_DMD}
\input{sections/connection_to_existing_methods}

\input{sections/numerical}
\input{sections/discussion}


\section*{Author Declarations}
\subsection*{Conflict of Interest}
The authors have no conflicts to disclose.

\subsection*{Author Contributions}
\textbf{Yoav Harris:} Conceptualization (equal); Formal analysis (lead); Investigation (lead); Methodology (lead); Software (lead); Validation (lead); Visualization (lead); Writing -- original draft (lead); Writing -- review \& editing (equal).
\textbf{Hadas Benisty:} Conceptualization (equal); Funding acquisition (equal); Supervision (equal); Writing -- review \& editing (equal).
\textbf{Ronen Talmon:} Conceptualization (equal); Funding acquisition (equal); Supervision (equal); Writing -- review \& editing (equal).

\input{sections/data_availability}

\section*{Appendices}

\appendix
\input{appendices/wide_regime}
\input{appendices/choose_m_singular_vectors}
\input{appendices/bound_on_subspace_deviation}
\input{appendices/concise_spatiotemporal_coupling}
\input{appendices/block_companion}
\input{appendices/fixed-eigenvalue-kv-fit}
\input{appendices/complexity}
\input{appendices/resdmd_in_our_setting}

\bibliographystyle{aipnum4-2}
\bibliography{bibliography/citations}

\end{document}

%% file: acronyms.tex
\DeclareAcronym{DMD}{
  short = DMD,
  long  = Dynamic Mode Decomposition
}
\DeclareAcronym{DCDMD}{
  short = DC--DMD,
  long  = Delay-Coordinates Dynamic Mode Decomposition
}
\DeclareAcronym{EDMD}{
  short = EDMD,
  long  = Extended Dynamic Mode Decomposition
}
\DeclareAcronym{SVD}{
  short = SVD,
  long  = Singular Value Decomposition
}
\DeclareAcronym{SNR}{
  short = SNR,
  long  = Signal-to-Noise Ratio
}
\DeclareAcronym{KV}{
  short = KV,
  long  = Kronecker--Vandermonde
}
\DeclareAcronym{ESR}{
  short = ESR,
  long  = Estimated-Subspace Residual
}
\DeclareAcronym{SSR}{
  short = SSR,
  long  = Signal-Subspace Residual
}
\DeclareAcronym{BIC}{
  short = BIC,
  long  = Bayesian Information Criterion
}
\DeclareAcronym{MDL}{
  short = MDL,
  long  = Minimum Description Length
}
\DeclareAcronym{STC}{
  short = STC,
  long  = Spatiotemporal Consistency
}
\DeclareAcronym{FEKVF}{
  short = FEKVF,
  long  = Fixed-Eigenvalue KV Fit
}

%% file: sections/introduction.tex
\section{Introduction}

High-dimensional time series arise in a wide range of physical and biological systems, including fluid flows \cite{Schmid2010,BerkoozHolmesLumley1993}, neural recordings \cite{CunninghamYu2014,Churchland2012}, climate dynamics \cite{Hannachi2007,ThompsonWallace2000}, and mechanical systems \cite{Kerschen2005,SaitoKuno2020}. Such time series are inevitably contaminated by noise and unmodeled effects. Although the data may involve many observables, the underlying dynamics are often governed by a relatively small number of coherent components. A central challenge in data-driven dynamical analysis is therefore to extract these latent components and characterize their temporal behavior directly from observed trajectories \cite{BennerGugercinWillcox2015}.

\Ac{DMD} is a widely used framework for this purpose. Originally introduced in the context of fluid dynamics, \ac{DMD} approximates the observed temporal evolution by a best-fit linear operator learned from sampled data and interprets its eigenpairs as spatiotemporal building blocks of the underlying dynamics \cite{Schmid2010,Tu2014,Kutz2016}. The associated eigenvalues encode oscillatory and growth or decay behavior, while the corresponding modes, which are constructed from the eigenvectors, encode spatial structure. Variants of \ac{DMD} and related operator-theoretic methods have since been applied across many scientific and engineering domains \cite{Rowley2009,Budisic2012,Colbrook2023}.

From an operator-theoretic perspective, \ac{DMD} can be interpreted as a finite-dimensional, data-driven approximation of the Koopman operator, which provides a linear representation of nonlinear dynamics through its action on observables \cite{Mezic2005,Mezic2013,WilliamsKevrekidisRowley2015,WilliamsHematiDawsonKevrekidisRowley2016}. This connection becomes particularly natural under lifting strategies that enrich the observable space and expose additional spectral content.

A common enhancement of \ac{DMD} is the use of \emph{delay coordinates}, where snapshot matrices are formed by stacking time-delayed samples into each column. This approach provides a systematic means of expanding the space of observables, improving spectral resolution, and strengthening the link between data-driven linear models and Koopman operator theory \cite{DasGiannakis2019,PanDuraisamyDelay2020,DasGiannakis2020}. Delay coordinates originate in nonlinear dynamics and attractor reconstruction \cite{Takens1981}, and are closely related to classical subspace identification, Hankel-based techniques, and matrix pencil methods \cite{JuangPappa1985,HuaSarkar1990}. In the context of \ac{DMD}, delay-coordinates formulations, introduced as Hankel-\ac{DMD} by Arbabi and Mezi\'c~\cite{ArbabiMezic2017} (see also the subsequent mathematical analysis in~\cite{Mezic2022}) and also referred to as \ac{DCDMD}, have been shown to enhance mode separation and recover more accurate spectral information from short or noisy time series~\cite{Tu2014,LeClaincheVega2017,PanArnoldMedabalimiDuraisamy2021,Schmid2022}.
One concrete illustration is the standing-wave case, where a single real oscillation is generated by a pair of complex-conjugate components. In this setting, standard DMD (without delay coordinates) often splits the oscillation into two separate modes whose individual spatial patterns are not physically meaningful, whereas \ac{DCDMD} recovers the correct oscillatory structure \cite{Bronstein2022}.
More generally, standard \ac{DMD} can fail whenever distinct oscillatory components have linearly dependent spatial patterns. This includes low-dimensional sensing regimes, where the number of active oscillatory components exceeds the number of spatial measurements. Delay coordinates can separate such components by replacing spatial snapshots with spatiotemporal patterns. This regime arises in a wide range of practical applications, including radar and array signal processing, communications, neuroscience, biomedical signal analysis, and localization. Some of these applications fall within standard \ac{DMD} usage, while others are conventionally addressed via closely related Hankel- and delay-based spectral estimators such as matrix pencil, ESPRIT, and Prony-type methods, to which the mode-selection scores developed here apply directly.

In the presence of noise, the snapshot matrices formed in \ac{DCDMD} are effectively full rank, and hence, practical implementations typically work with truncation rank that only approximates the unknown true model order. When the model order is overestimated, the computed decomposition necessarily includes additional components that have no counterpart in the underlying dynamics. 

The problem of identifying which \ac{DMD} components reflect true dynamics, and which are spurious artifacts of noise and order overestimation, has therefore received considerable attention. Existing approaches include spectral heuristics \cite{Tu2014,Kutz2016}, likelihood- and information-theoretic criteria \cite{Rissanen1978,Schwarz1978,WaxKailath1985}, robustness-based methods \cite{SashidharKutz2022,DawsonHematiWilliamsRowley2016}, structure-aware post hoc tests that exploit delay-induced spatiotemporal structure \cite{Bronstein2022}, the general-purpose eigenpair residual diagnostic ResDMD~\cite{ColbrookTownsend2024}, and Koopman Reduced-Order Modeling (KROM)~\cite{MohrFonoberovaMezic2025}.
Importantly, in delay-coordinate settings, detecting the model order becomes more challenging. As we show in the sequel, delay coordinates induce spatiotemporal organization in the lifted space, where the spurious additional components appear coherent across delays, and therefore, become more difficult to separate from true modes than in standard (non-delayed) \ac{DMD}; as a result, common heuristics such as singular-value gap rules, eigenvalue screening, or mode-norm rankings become unreliable.

Despite the extensive body of work, order detection in \ac{DCDMD} lacks a principled, perturbation-robust definition of what distinguishes a true dynamical component from a spurious one. Existing approaches largely treat spurious modes as residual artifacts of noise or overfitting, without explicitly characterizing their structural relationship to the signal-dominated subspaces induced by the data. As a result, selection rules are typically formulated in terms of spectral magnitude, likelihood penalties, or empirical stability, rather than in terms of intrinsic geometric properties of the recovered modes.

In this paper, we define mode selection in \ac{DCDMD} through subspace geometry. True modes are those that concentrate their energy within a low-dimensional signal subspace, while spurious modes necessarily retain non-negligible components outside any moderate overestimate of this subspace. This geometric viewpoint yields a perturbation-robust definition of true and spurious modes and leads to practical per-mode scores. In the low-dimensional regimes emphasized below, the residual geometry is most useful when paired with the delay-induced structure of the modes. We therefore treat mode selection as a scoring problem: each computed mode receives a per-mode score, and the hard true-vs-spurious decision (threshold or clustering) is a downstream choice independent of score design.

Building on this geometric framework, we develop two complementary mode-selection scores. The first translates the signal-subspace concentration principle into a fully data-driven residual score by constructing a proxy for the signal subspace and quantifying each mode's deviation from it.

The second score builds on a structural analysis of delay-coordinate modes. Prior work has identified a \ac{KV} structure in true \ac{DCDMD} modes under explicit exponential signal models~\cite{Bronstein2022}. Here, we show that this structure does not arise from the signal model itself, but is instead induced by the delay-coordinates and the block-companion structure of the least-squares propagator. In particular, we prove that all DMD modes, true and spurious, inherit an intrinsic \ac{KV} organization, and that the distinguishing feature of true modes is not the presence of \ac{KV} structure, but the degree to which they conform to it. This observation shifts the interpretation of \ac{KV} structure from a binary signature of true components to a continuous structural property governed by subspace geometry. The deviation from ideal \ac{KV} structure is shown to be directly controlled by the geometric residual introduced earlier, yielding a principled structural score for mode selection.

Finally, our analysis provides a principled explanation for the empirical behavior of commonly used mode-energy and mode-norm heuristics. By decomposing the exact-mode norm into eigenvalue-dependent and residual-dependent contributions, we show that delay coordinates fundamentally alter the spectral statistics of spurious modes. In particular, as the embedding length grows, spurious eigenvalues concentrate near the unit circle, rendering magnitude- and norm-based rankings unreliable or even misleading in weakly damped regimes. This clarifies when such heuristics can be expected to fail, and why residual- and structure-based criteria provide more robust alternatives in delay-coordinate settings.

We validate the proposed scores through numerical experiments evaluated by precision--recall AUC of the true-vs-spurious ranking task. The experiments are concentrated in the regime where the spatial dimension $D$ is small --- precisely the regime in which the lag-axis \ac{KV} structure carries most of the identifying signal --- and span sweeps over noise level, spectral separation, damping, and amplitude heterogeneity. {Across most of this regime, the proposed geometric and structural scores outperform two leading baselines: the spatiotemporal coupling score of Bronstein \emph{et al.}~\cite{Bronstein2022} and the residual score of ResDMD~\cite{ColbrookTownsend2024}}. We additionally characterize, empirically, the behavior of spurious eigenvalues under delay coordinates.

The remainder of this paper is organized as follows. Section~\ref{sec:problem} formulates the setting, including the signal and delay-coordinate model, the \ac{DMD} construction, the signal subspace, and the geometric mode-selection viewpoint. Section~\ref{sec:related_work} reviews related mode-selection methods. Sections~\ref{sec:mode_geometry} and~\ref{sec:kv_structure} develop the geometric and structural scores, respectively, and Section~\ref{sec:connection_to_existing_methods} interprets common heuristics through this framework. Section~\ref{sec:numerical_experiments} presents numerical results in the small-$D$ regime, and Section~\ref{sec:discussion} concludes with a summary and implications.

%% file: sections/problem_setting.tex
\section{Problem formulation}\label{sec:problem}
In this section, we formalize the \ac{DCDMD} setting in a form suitable for geometric analysis. We introduce the signal model, the associated signal subspace, and the rank-constrained DMD construction, emphasizing the subspace relationships that underlie the distinction between true and spurious dynamical modes and our subsequent theoretical analysis and proposed algorithms.

\subsection{System and signal model}\label{subsec:sys_and_data}
We observe a signal in discrete time given by:
\begin{equation}
    \bm{x}_k = \bm{s}_k + \bm{n}_k \in \mathbb{C}^D,
    \label{eq:measurement_model}
\end{equation}
where $\bm{n}_k$ represents noise, $\bm{s}_k$ follows a deterministic time-evolution rule given by:
\begin{equation}
    \bm{s}_{k+1} = f( \bm{s}_k ) , \qquad \bm{s}_k \in \mathbb{C}^D,
\label{eq:deterministic_evolution}    
\end{equation}
and $k$ is the discrete time index.
Given the signal model \eqref{eq:measurement_model} and the latent dynamical evolution \eqref{eq:deterministic_evolution}, our objective is to recover the latent dynamics from a finite snapshot of $N$ consecutive noisy observation samples $\{\bm{x}_k\}_{k=0}^{N-1}$.

We pursue two complementary outcomes: (i) recover the map $f$ in \eqref{eq:deterministic_evolution} when it is identifiable from the data; or (ii) when full identification is not possible, obtain a faithful characterization of the dynamics, such as its dominant frequencies, growth or decay rates, and spatial structures. Such representations support tasks including denoising, by separating $\bm{s}_k$ from $\bm{n}_k$, and prediction, by forecasting $\bm{s}_{k+h}$ beyond the observation snapshot.

\subsection[Delay-coordinates DMD]{Delay-coordinates \ac{DMD}}
\label{subsec:delay_dmd}\label{subsec:delay_coords}\label{subsec:dmd_construction}

\ac{DMD} and related methods commonly use delay-coordinates where each delay vector stacks $L$ consecutive samples:
\begin{equation}
\widetilde{\bm{x}}_k
    \;\coloneqq\;
    \begin{bmatrix}
        \bm{x}_k \\ \bm{x}_{k+1} \\ \vdots \\ \bm{x}_{k+L-1}
    \end{bmatrix}
    \in\mathbb{C}^{DL},
\qquad k = 0,\dots,N-L .
\label{eq:delay_vector}
\end{equation}

We define $\widetilde{\bm{s}}_k$ and $\widetilde{\bm{n}}_k$ similarly. Collecting all the delay vectors yields
\begin{equation}
\bm{X}^{(L)} \coloneqq
\begin{bmatrix}
\widetilde{\bm{x}}_0 &  \cdots & \widetilde{\bm{x}}_{N-L}
\end{bmatrix} = \bm{S}^{(L)} + \bm{E}^{(L)}
\in \mathbb{C}^{DL \times (N-L+1)} .
\label{eq:traj_mat_L_def}
\end{equation}
In the remainder of this paper, unless explicitly stated otherwise, we fix $L\ge 1$ and, for brevity, we omit the superscript and write $\bm{X}, \bm{S}, \bm{E}$ for $\bm{X}^{(L)}, \bm{S}^{(L)}, \bm{E}^{(L)}$.

Classical delay–embedding theory shows that, under generic conditions,
a delay map embeds the underlying attractor \cite{Takens1981}, providing
a standard basis for using delay-coordinates to enrich observed dynamics.
In practice, delay-coordinates are widely used to improve modal separation
and spectral resolution in data-driven analysis of dynamical systems
\cite{JuangPappa1985,ArbabiMezic2017,Tu2014,LeClaincheVega2017}.

\Ac{DMD} approximates the temporal evolution map $f$ from \eqref{eq:deterministic_evolution} by a linear propagator:
\begin{equation}
    \bm{s}_{k+1} = \bm{A}\,\bm{s}_k.
\label{eq:linear_signal_evolution}
\end{equation}
Any trajectory of~\eqref{eq:linear_signal_evolution} admits a finite modal representation of the form:
\begin{equation}
\bm{s}_k \;=\; \sum_{j=1}^{m} b_j\,\bm{\phi}_j\,\lambda_j^{\,k}
\;=\; \bm{\Phi}\,\bm{\Lambda}^{k}\bm{b},
\label{eq:explicit_signal_structure}
\end{equation}
where $(\lambda_j,\bm{\phi}_j)$ are eigenpairs of $\bm{A}$, $b_j$ are the amplitudes determined by the initial condition $\bm{s}_0$, and $m$ is the model order. Throughout, we parameterize the complex eigenvalues as
\begin{equation}
\lambda_j \;=\; \rho_j\,\mathrm{e}^{\mathrm{i}\theta_j},
\qquad \rho_j \in [0,1],
\label{eq:eigenvalue_parameterization}
\end{equation}
where $\rho_j$ are magnitudes, $\theta_j$ are phases, and $\rho_j \le 1$ ensures bounded temporal behavior.
The triplets $(\rho_j,\theta_j,\bm{\phi}_j)$ together with $b_j$ are the signal parameters; throughout, we refer to them as the true underlying signal components.
This modal form holds exactly when $\bm{A}$ is diagonalizable. In the non-diagonalizable case, polynomial-exponential dynamics are typically represented in finite-data DMD by clusters of nearby simple eigenvalues; see \cite{Tu2014,Kutz2016,ArbabiMezic2017}.

Since the delay-coordinates embedding in \eqref{eq:delay_vector} is linear, \eqref{eq:explicit_signal_structure} lifts to
\begin{equation}
\widetilde{\bm{s}}_k \;=\; \sum_{j=1}^{m} b_j\,\widetilde{\bm{\phi}}_j\,\lambda_j^{\,k},
\qquad
\widetilde{\bm{\phi}}_j \coloneqq \bm{v}_L(\lambda_j) \otimes \bm{\phi}_j \;\in\; \mathbb{C}^{DL},
\label{eq:lifted_signal_structure}
\end{equation}
with $\bm{v}_L(\lambda) = [1,\lambda,\ldots,\lambda^{L-1}]^\top$ the Vandermonde vector (Kronecker form of the construction in Bronstein \emph{et al.}~\cite{Bronstein2022}; see App.~\ref{app:data_side_kv}). We refer to $\widetilde{\bm{\phi}}_j$ as the \emph{lifted underlying modes}. The linear evolution of $\bm{s}_k$ carries over to the embedded vectors, motivating the approximation
\begin{equation}
 \widetilde{\bm{x}}_{k+1}\approx \widetilde{\bm{A}}\,\widetilde{\bm{x}}_{k},
\label{eq:linear_dynamic_approximation}
\end{equation}
where $\widetilde{\bm{A}} \in \mathbb{C}^{DL \times DL}$ is the lifted-space operator whose structure is examined in Sec.~\ref{sec:kv_structure}.
To estimate a one-step propagator, we pair each delay vector with its immediate successor and form two snapshot matrices
\begin{equation}
\begin{split}
\bm{X}_0 &=
\begin{bmatrix}
\widetilde{\bm{x}}_0 & \widetilde{\bm{x}}_1 & \cdots & \widetilde{\bm{x}}_{N-L-1}
\end{bmatrix},
\\
\bm{X}_1 &=
\begin{bmatrix}
\widetilde{\bm{x}}_1 & \widetilde{\bm{x}}_2 & \cdots & \widetilde{\bm{x}}_{N-L}
\end{bmatrix},
\end{split}
\label{eq:x0_x1}
\end{equation}
both of size $DL \times (N-L)$. An analogous construction defines $\bm{S}_0, \bm{S}_1, \bm{E}_0, \bm{E}_1$, following \eqref{eq:traj_mat_L_def}. 

The paired snapshot matrices satisfy $\operatorname{rank}(\bm{S}_0)=\operatorname{rank}(\bm{S}_1)=m$, since all their columns lie in $\mathrm{span}\{\widetilde{\bm{\phi}}_1,\dots,\widetilde{\bm{\phi}}_m\}$.
We therefore use “order” and “signal rank” interchangeably.
As is typical in a wide range of systems, the number of underlying dynamical components is substantially smaller than the number of observables, so we adopt the standard assumption that
\[
m \ll \min(N\!-\!L,\,DL).
\]
The maximal algebraic rank of the signal matrices is $\min(N\!-\!L,\,DL)$. We focus on the common regime $N\!-\!L<DL$ (the complementary case $DL<N\!-\!L$ is less common and discussed in Appendix~\ref{app:wide_matrix}).

Although the linear evolution is postulated for the latent signal, $\bm{S}_1=\bm{A}\bm{S}_0$, in practice we estimate $\bm{A}$ from the noisy observations by solving $\bm{X}_1 \approx \bm{A}\bm{X}_0$.

Noise affects the estimation process in two coupled ways: additive noise both perturbs the signal subspace and introduces directions that are not aligned with any low-dimensional signal geometry, thereby complicating the separation between signal-consistent and noise-consistent components.
In particular, by our signal model \eqref{eq:measurement_model}, we have $\bm{X}=\bm{S}+\bm{E}$, where the additive noise $\bm{E}$ obscures the low-rank structure of the signal $\bm{S}$, rendering $\bm{X}_0$ and $\bm{X}_1$ effectively full rank. Consequently, the true order $m$ cannot be inferred from the ranks of the snapshot matrices alone. 
Because the true order $m$ is unknown, we compute DMD at a working rank $M \ge m$, typically chosen strictly larger than $m$. This overestimation introduces additional components that have no signal counterpart and must be separated from the true modes.

Whether $M$ coincides with $m$, and how deviations between them manifest, lies at the core of this work and is formalized as a \emph{mode-selection} problem in Section~\ref{subsec:mode_selection_problem}: among the $M$ computed modes, identify which approximate the $m$ underlying true modes and which are spurious. An estimate of $m$ (order detection) then follows as a downstream by-product.

Given the paired snapshot matrices $(\bm{X}_0,\bm{X}_1)$, a natural unconstrained least–squares solution, via the Moore–Penrose pseudoinverse \cite{Penrose1956, BarataHussein2012}, is
\begin{equation}
\label{eq:moore_penrose_ls}
    \widehat{\bm{A}}_{\mathrm{MP}}=\bm{X}_1\bm{X}_0^{\dagger}
\end{equation}
However, in the presence of noise, this produces a (numerically) full-rank operator that absorbs high-variance directions and thereby \emph{models noise as dynamics}. This undermines the identifiability of the true components.
The low-rank assumption motivates solving a rank-constrained least squares:
\begin{equation}
\widehat{\bm{A}} = 
\arg\min_{\substack{\bm{A}\in\mathbb{C}^{DL\times DL} \\ \operatorname{rank}(\bm{A})=M}}
\;\|\bm{X}_1-\bm{A}\bm{X}_0\|_F^2
\label{eq:rank_constrained_least_squares_objective}
\end{equation}

Broadly, \ac{DMD} extracts the eigenpairs of the estimated propagator $\widehat{\bm{A}}$ from the rank-constrained least squares \eqref{eq:rank_constrained_least_squares_objective}.
Constraining the rank in \eqref{eq:rank_constrained_least_squares_objective} enforces the low-rank signal model, restricting the fit to an $M$-dimensional subspace.
Geometrically, this restriction confines the fitted dynamics to an M-dimensional truncation subspace, which ideally contains the signal subspace $\mathcal S$, but may also include additional noise-aligned directions.
This yields an explicit spatiotemporal representation, where the eigenvalues capture temporal evolution while the eigenvectors capture spatial structure. In this form, the signal can be modeled and reconstructed directly, rather than merely advanced one step by~$\widehat{\bm{A}}$.

The practical solution to the rank-constrained least-squares problem \eqref{eq:rank_constrained_least_squares_objective} uses the truncated \ac{SVD} of $\bm{X}_0$ (with truncation rank $M$):
\begin{equation}
\begin{gathered}
\bm{X}_0 \;\approx\; \bm{U}_M\,\bm{\Sigma}_M\,\bm{V}_M^{\mathrm{H}}, \\[0.5em]
\bm{U}_M \in \mathbb{C}^{DL\times M},\;\;
\bm{\Sigma}_M \in \mathbb{R}^{M\times M},\;\;
\bm{V}_M \in \mathbb{C}^{(N-L)\times M}.
\end{gathered}
\label{eq:svd_x0}
\end{equation}
Recall that the columns of SVD matrices are ordered by the (real and positive) singular values.
These matrices are used to construct the truncated propagator:
\begin{equation}
\bm{A}_M \;\coloneqq\; \bm{U}_M^{\mathrm{H}}\,\bm{X}_1\,\bm{V}_M\,\bm{\Sigma}_M^{-1} 
\;\in\; \mathbb{C}^{M\times M}.
\label{eq:reduced_propagator}
\end{equation}
Since exact spectral defectiveness is structurally unstable under finite data and perturbations \cite{Tu2014,Kutz2016,ArbabiMezic2017}, $\bm{A}_M$ generically has a simple spectrum and admits an eigenvalue decomposition:
\begin{equation*}
\bm{A}_M = \bm{W}\,\widehat{\bm{\Lambda}}\,\bm{W}^{-1}, \end{equation*}
where
\begin{equation*}
\widehat{\bm{\Lambda}} 
= \operatorname{diag}(\widehat{\lambda}_1,\ldots,\widehat{\lambda}_M)
\in \mathbb{C}^{M\times M}, 
\end{equation*}
is a diagonal matrix consisting of the eigenvalues $\hat{\lambda}_j$, and
\begin{equation*}
\bm{W} 
= [\,\bm{w}_1\;\cdots\;\bm{w}_M\,]
\in \mathbb{C}^{M\times M}
\end{equation*}
is a matrix consisting of the eigenvectors $\bm{w}_j$.

In the DMD literature, there exist two standard methods for mapping a reduced eigenvector $\bm{w}_j \in \mathbb{C}^M$ to a \emph{mode} in the signal space $\mathbb{C}^{DL}$, a mapping commonly referred to as \emph{lifting}.

\emph{Projected \ac{DMD}} defines the projected DMD modes using the truncated SVD matrix $\bm{U}_M$ as follows:
\begin{equation}
\label{eq:proj_mode_def}
\widehat{\bm{\phi}}_j^{\mathrm{p}} \coloneqq  \bm{U}_M \bm{w}_j.
\end{equation}
By definition, $\widehat{\bm{\phi}}_j^{\mathrm{p}} \in \mathcal{U}_M$, where $\mathcal{U}_M$ denotes the column space of $\bm{U}_M$:
\begin{equation}
\label{eq:UM_colspace}
\mathcal{U}_M \coloneqq \operatorname{col}(\bm{U}_M).
\end{equation}

We note that since $\|\bm{w}_j\|_2 = 1$ and $\bm{U}_M$ has orthonormal columns, by their construction, projected modes are unit vectors:
\[\|\bm{\phi}_j^{\mathrm{p}}\|_2 = \|\bm{U}_M \bm{w}_j\|_2 = 1 .
\]

\emph{Exact \ac{DMD}} uses the (untruncated) snapshot matrix $\bm{X}_1$ and defines the exact DMD modes as follows:
\begin{equation}
\label{eq:exact_mode_def}
\widehat{\bm{\phi}}_j^{\mathrm{e}} \coloneqq  \bm{X}_1 \bm{V}_M \bm{\Sigma}_M^{-1}\bm{w}_j.
\end{equation}
Here,  $\widehat{\bm{\phi}}_j^{\mathrm{e}} \in \operatorname{col}(\bm{X}_1)$, implying that the exact DMD modes are not restricted to any subspace of rank $M$.

Several remarks are due at this point. First, we note that a mode without a superscript, $\widehat{\bm{\phi}}_j$, may denote either the projected \ac{DMD} mode $\widehat{\bm{\phi}}^{\,p}_j$ or the exact \ac{DMD} mode $\widehat{\bm{\phi}}^{\,e}_j$. Therefore, statements written with $\widehat{\bm{\phi}}_j$ hold for both types.

Second, following the definition of $\bm{A}_M$ in~\eqref{eq:reduced_propagator}, the projected and exact modes differ only by projection and scaling, and their explicit relationship is given by
\begin{equation}
\label{eq:proj_identity}
\bm{P}_{\mathcal{U}_M}\,\widehat{\bm{\phi}}^{\,\mathrm{e}}_j \;=\; \widehat{\lambda}_j\,\widehat{\bm{\phi}}^{\,\mathrm{p}}_j,
\end{equation}
where
\begin{equation}
    \bm{P}_{\mathcal{U}_M} \coloneqq \bm{U}_M \bm{U}_M^{\mathrm{H}}
\label{eq:P_U_M}
\end{equation}
is the orthogonal projector onto $\mathcal{U}_M$.

Third, in the ideal noiseless case using correct order ($M=m$), the projected \ac{DMD} and exact \ac{DMD} coincide: they yield the same eigenvalues and the same modes. 

Fourth, we emphasize the distinction between the \emph{underlying signal
eigenvalues} $\{\lambda_j\}_{j=1}^{m}$ and \emph{underlying signal modes}
$\{\bm{\phi}_j\}_{j=1}^{m}$ in~\eqref{eq:explicit_signal_structure}, and the
eigenvalues and modes computed by \ac{DMD}. The latter are estimates of the former.
We collect the modes (from either lifting scheme) into a matrix 
\[
\widehat{\bm{\Phi}} = [\,\widehat{\bm{\phi}}_{1}\;\cdots\;\widehat{\bm{\phi}}_{M}\,] \in \mathbb{C}^{DL \times M}.
\]
The amplitudes $\widehat{\bm{b}}$ are determined by enforcing the initial condition,
i.e., by solving $\widehat{\bm{\Phi}}\widehat{\bm{b}} \approx \widetilde{\bm{x}}_0$.
Then, the order-$M$ approximation of the delay-coordinates samples is 
\begin{equation}
\widehat{\bm{x}}_{k} \;=\;
\widehat{\bm{\Phi}}\,
\widehat{\bm{\Lambda}}^{\,k}\,
\widehat{\bm{b}},
\qquad k=0,\dots,N-L .
\label{eq:reconstruction}
\end{equation}

Note that \eqref{eq:reconstruction} extends naturally to prediction: for any integer $h=1,2,\dots$, we have
\[
\widehat{\bm{x}}_{k+h}
\;=\;
\widehat{\bm{\Phi}}\,
\widehat{\bm{\Lambda}}^{\,k+h}\,
\widehat{\bm{b}},
\qquad k \geq N-L,
\]
where $\widehat{\bm{x}}_{k+h}$ is a prediction in the delay-coordinates space. 
The corresponding prediction in the original signal space is obtained by extracting its first $D$ entries.

\subsection{Signal subspace and the mode-selection problem}
\label{subsec:mode_selection_problem}

Central to our geometric framework is the \emph{signal subspace}
\begin{equation}
\mathcal S \coloneqq \operatorname{span}\{\widetilde{\bm{\phi}}_1,\dots,\widetilde{\bm{\phi}}_m\} \;\subset\; \mathbb{C}^{DL},
\label{eq:signal_subspace_def}
\end{equation}
which contains the structure of the underlying low-dimensional dynamics, where $\widetilde{\bm{\phi}}_1,\dots,\widetilde{\bm{\phi}}_m$ are the lifted underlying modes~\eqref{eq:lifted_signal_structure}.
By construction, the noiseless snapshot matrices satisfy
\begin{equation}
\operatorname{col}(\bm{S}_0) = \operatorname{col}(\bm{S}_1) = \mathcal S.
\end{equation}

We assume that the noise level is sufficiently low such that the span of the principal components of the noisy snapshot matrix $\mathcal{U}_M$ defined in \eqref{eq:UM_colspace} approximately contains the signal subspace $\mathcal{S}$. 

We define $\mathcal{U}_m = \operatorname{col}(\bm{U}_m)$, where canonically $\bm{U}_m$ contains the leading $m$ left singular vectors, i.e., the $m$ leftmost columns of $\bm{U}_M$. This canonical choice is adopted only for conceptual convenience and has no practical bearing on the developments below. Alternative selections of $m$ vectors from $\bm{U}_M$ that optimally approximate the signal subspace are discussed in Appendix~\ref{app:choose_m_vecs_out_of_M}.
We make the following quantitative assumption regarding how well $\mathcal{U}_m$ is aligned with the signal subspace $\mathcal{S}$.
We quantify this alignment by the spectral norm of the difference between orthogonal projectors, since it is basis-invariant, and has a direct principal-angle interpretation.

\begin{assumption}[Signal-subspace proximity]\label{ass:subspace_distance_bound}
There exists a small $\eta\in[0,1)$ such that\begin{equation}\label{eq:eta_bound}
\bigl\|\bm P_{\mathcal S} - \bm P_{\mathcal U_m}\bigr\|_2 \le \eta.
\end{equation}
\end{assumption}
Since $\dim\mathcal S=\dim\mathcal U_m=m$, \eqref{eq:eta_bound} is equivalent to
\[
\theta_{\max}(\mathcal S,\mathcal U_m) \le \arcsin\eta,
\]
where $\theta_{\max}$ is the largest principal angle.

Assumption~\ref{ass:subspace_distance_bound} formalizes the requirement that, under moderate noise, the $m$-dimensional truncation subspace $\mathcal{U}_m$ remains sufficiently aligned with the true signal subspace $\mathcal S$, namely,
\begin{equation}\label{eq:subspace_hierarchy}
\mathcal S \;\approx\; \mathcal U_m \;\subseteq\; \mathcal U_M,
\end{equation}
where ``$\approx$'' refers to the projector-distance bound
\eqref{eq:eta_bound}, and the inclusion holds since $\mathcal U_m$ is spanned by
a subset of the columns of $\bm U_M$.
This proximity condition ensures that signal-consistent modes can be distinguished geometrically from spurious components.
Heuristically, the approximation of $\mathcal{S}$ by $\mathcal{U}_m$ improves (suggesting the existence of smaller $\eta$) with higher \ac{SNR},
better spectral separation, better spatial conditioning, and larger embedding
dimension or window length; see Appendix~\ref{app:signal_subspace_identifiability}.
In what follows, $\mathcal U_m$ serves as an ideal reference subspace, while $\mathcal U_M$ provides the practical, data-driven proxy used in mode selection.

In summary, the true order $m$ is unknown, and DMD computed at an overestimated rank $M \ge m$ produces both signal-consistent and spurious components. The central problem is therefore to identify, among the $M$ computed modes, those that are geometrically consistent with the signal subspace. We approach this as a scoring problem: each computed mode is assigned a score that ranks it as true or spurious; any hard decision can then be obtained from these scores by a downstream rule such as thresholding (as in ResDMD~\cite{ColbrookTownsend2024}) or clustering (as in Bronstein \emph{et al.}~\cite{Bronstein2022}), with the estimated order $\widehat m$ recoverable as the number of selected modes. Since the decision rule is orthogonal to score design, what matters is the separation power of the score itself, which we evaluate in Sec.~\ref{sec:numerical_experiments} by the \emph{precision--recall AUC} of the true-vs-spurious ranking (i.e., across all thresholds).\label{subsec:score_to_decision} This post hoc formulation provides the structural basis for the mode-selection methods developed in the sequel.

%% file: sections/related_work.tex
\section{Related work on mode selection and order detection}
\label{sec:related_work}

Existing approaches for mode selection and order detection in \ac{DMD} can be grouped into four broad families. For comprehensive overviews of \ac{DMD} variants, selection heuristics, and algorithmic foundations, see\cite{Schmid2022,Colbrook2023}.

The first group consists of \emph{spectral and likelihood-based mode selection} methods. Spectral heuristics are widely used for order detection in \ac{DMD} and for selecting the truncation rank, using gap- or threshold-based rules applied to the singular values of the snapshot matrix \cite{Tu2014,Kutz2016}. In the context of delay-coordinates, analogous rules are applied to the singular values of the corresponding block-Hankel snapshot matrices \cite{LeClaincheVega2017}. Eigenvalue screening rules based on magnitude or phase are also commonly employed in practice \cite{Kutz2016}. Closely related approaches formulate order detection as penalized likelihood model selection, trading off fit quality with a complexity penalty, for example via \ac{MDL} or \ac{BIC}~\cite{Rissanen1978,Schwarz1978,StoicaSelen2004,WaxKailath1985}. Within the \ac{DMD} framework, likelihood-based selection has been pursued through penalized optimization formulations~\cite{Jovanovic2014,ProctorBruntonKutz2016} or Bayesian
inference approaches \cite{PanDuraisamy2020}.

The second group comprises \emph{robustness-based mode selection} strategies, which retain components that persist under data resampling, held-out refits, or other perturbations of the snapshots. A prominent example is bagging optimized DMD (BOP-DMD), which fits ensembles across randomized data subsets and selects modes based on their empirical consistency \cite{SashidharKutz2022}. Other stability-based methods assess mode validity using diagnostics such as time-reversal consistency~\cite{DawsonHematiWilliamsRowley2016}.

The third group includes \emph{post hoc mode-selection}
methods, which compute \ac{DMD} at an overestimated truncation rank and then perform mode selection in a post-processing step.
In contrast to rank estimation or eigenvalue-centric selection methods, these approaches are mode-centric and exploit structural properties of the computed modes.
In settings involving delay-coordinates, representative examples test for \ac{KV} structure across delays in the resulting delay-coordinates modes
\cite{Bronstein2022}. Related post-processing strategies have also been
developed in closely connected subspace and matrix-pencil settings~\cite{SegmanAmarTalmon2025,SegmanAmarTalmon2025Plus}.

The fourth group consists of \emph{amplitude-based heuristics}, which rank modes by simple magnitude proxies such as Euclidean norms or reconstruction amplitudes~\cite{Schmid2010,Tu2014}. These should be understood as practical magnitude proxies, rather than as likelihood-based models. The behavior of these proxies under delay-coordinates is analyzed in Sec.~\ref{subsec:mode_norm_analysis}.

Two recent approaches do not fit neatly into the above families. ResDMD~\cite{ColbrookTownsend2024} is a post hoc method that ranks eigenpairs by an approximated residual with respect to the ideal Koopman operator, providing a general-purpose diagnostic that is not tailored to delay-coordinate structure; related theory on learning Koopman spectra from trajectory data was developed by Colbrook \emph{et al.}~\cite{ColbrookMezicStepanenko2024}. The Koopman Reduced-Order Modeling (KROM) framework of Mohr \emph{et al.}~\cite{MohrFonoberovaMezic2025} combines amplitude ordering with a likelihood-style residual test: it orders modes by amplitude and chooses the retained order by testing whether a projected reconstruction error is better explained as Gaussian noise than as unresolved structured dynamics, yielding an order-detection rule and prediction confidence bounds rather than a per-mode selection rule.

The above approaches address order detection and mode selection from complementary perspectives, but important gaps remain.
Spectral heuristics rely on global summaries that conflate eigenvalue scaling with subspace residual geometry. Information-criterion methods require refitting models across hypothesized orders, making them computationally expensive and sensitive to noise-model mismatch. Robustness-based approaches emphasize empirical stability but do not provide a structural interpretation of the recovered components or a principled definition of true versus spurious modes. Norm- and amplitude-based heuristics are simple to apply but lack a clear mechanistic justification under delay-coordinates.
Structure-aware post hoc methods explicitly formalize the expected organization of \emph{true} modes arising from delay-coordinates, but the corresponding geometric characterization of \emph{spurious} modes remains largely implicit. Spurious components are typically treated as unstructured residual artifacts rather than as objects with their own systematic geometry, and general-purpose residual diagnostics offer no delay-coordinate-specific account of them either. As a result, existing methods do not exploit the subspace and structural organization of spurious modes in their selection rules, leaving the geometric distinction between true and spurious components under delay-coordinates insufficiently understood.

%% file: sections/mode_geometry.tex
\section{Mode geometry, signal subspace residual, and true and spurious modes}
\label{sec:mode_geometry}

This section presents a geometric analysis of the \ac{DMD} modes and, based on it, derives a computable residual score for mode selection.

By definition \eqref{eq:signal_subspace_def}, the signal subspace $\mathcal{S}$ spans the $m$ lifted underlying modes $\{\widetilde{\bm{\phi}}_j\}_{j=1}^m$~\eqref{eq:lifted_signal_structure}.
However, in practice, $m$ is unknown and the underlying signal modes are accessed only through the $M>m$ computed \ac{DMD} modes $\{\widehat{\bm{\phi}}_j\}_{j=1}^M$. We now define a partition of the computed \ac{DMD} modes into \emph{true} and \emph{spurious} modes, where the true modes can be viewed as noisy estimators of the underlying signal modes.

\begin{definition}[Signal-subspace residual vector]
\label{def:ssr_vector}
For any vector $\bm{v}\in\mathbb{C}^{DL}$, the \ac{SSR}
vector is defined by
\begin{equation}
\bm r_{\mathcal S}(\bm v)
\;\coloneqq\;
(\bm{I}-\bm{P}_{\mathcal S})\,\bm{v},
\label{eq:ssr_vector_def}
\end{equation}
where $\bm{P}_{\mathcal S}$ is the orthogonal projector onto the signal subspace $\mathcal S$.
\end{definition}

We define the \emph{true} modes to be the $m$ modes with the smallest residual norms, and the remaining $M-m$ modes to be \emph{spurious}.

\begin{definition}[True and spurious modes]
\label{def:true_spurious_modes}
Given $M$ computed modes $\{\widehat{\bm{\phi}}_j\}_{j=1}^M$ (either exact or projected),
let $\mathcal J_{\mathrm{true}}\subset\{1,\dots,M\}$ be a set of $m$ indices attaining the
$m$ smallest values of
\[
\bigl\|\bm r_{\mathcal{S}}(\widehat{\bm{\phi}}_j)\bigr\|_2^2,\qquad j=1,\ldots,M,
\]
and~let $\mathcal J_{\mathrm{spur}} \coloneqq \{1,\dots,M\}\setminus\mathcal J_{\mathrm{true}}$.
For each $j\in\{1,\dots,M\}$, $\widehat{\bm{\phi}}_j$ is a \emph{true mode} if $j\in\mathcal J_{\mathrm{true}}$, and a \emph{spurious mode} otherwise.
\end{definition}

Definition \ref{def:true_spurious_modes} suggests that the modes can be classified as true or spurious by evaluating their residual energy outside the signal subspace $\mathcal{S}$.
However, since $\mathcal{S}$ cannot be directly observed, this definition serves only as an ideal reference for distinguishing between true and spurious modes.
To bridge this gap, we introduce a measurable statistic for the SSR.

By Assumption~\ref{ass:subspace_distance_bound}, the $m$-dimensional subspace $\mathcal U_m$ is close to
$\mathcal S$.
However, the true order $m$, and therefore $\mathcal U_m$, are unknown.
We postulate that $\mathcal U_M$, where $M>m$ is an arbitrary overestimate of the true order, can serve as a practical proxy for $\mathcal S$.
We therefore propose to measure residual energy with respect to $\mathcal U_M$ rather than
$\mathcal S$, using the orthogonal projector $\bm P_{\mathcal U_M}$.

\begin{definition}[Estimated-subspace residual vector]
\label{def:esr_vector}
For an exact \ac{DMD} mode $\widehat{\bm{\phi}}^{\,e}_j$, define the
\ac{ESR} vector by
\begin{equation}
\bm r_{\mathcal U_M}\bigl(\widehat{\bm{\phi}}^{\,e}_j\bigr)
\;\coloneqq\;
(\bm I-\bm P_{\mathcal U_M})\,\widehat{\bm{\phi}}^{\,e}_j,
\qquad j=1,\dots,M.
\label{eq:esr_vector_def}
\end{equation}
\end{definition}

We note that Definition~\ref{def:esr_vector} is only informative for exact
\ac{DMD} modes because the projected modes lie in $\mathcal U_M$ by construction, and so their residual is $\bm r_{\mathcal U_M}\bigl(\widehat{\bm{\phi}}^{\,p}_j\bigr) \equiv 0$.

Figure~\ref{fig:two_panel_residual} summarizes the geometry of the \ac{SSR} and the \ac{ESR} and highlights that \ac{ESR} captures the $\mathcal U_M^\perp$ component of \ac{SSR}.

\input{figures/residual_visualization}

We note that the term \emph{residual} is used in several distinct senses across the \ac{DMD} literature. Our usage is geometric: the residual of a \emph{mode} with respect to a designated \emph{subspace} (Definitions~\ref{def:ssr_vector},~\ref{def:esr_vector}). This differs from the \emph{eigenpair residual} of ResDMD~\cite{ColbrookTownsend2024}, which measures how well a candidate eigenpair satisfies the Koopman eigenrelation against a finite-dimensional approximation, and from the \emph{projection residual} of KROM~\cite{MohrFonoberovaMezic2025}, computed on partial-model reconstruction errors.

We define the ESR-energy score based on the squared norm of the \ac{ESR} vectors:
\begin{equation}\label{eq:esr_energy_def}
\mathcal R_j \;\coloneqq\; \bigl\|\bm r_{\mathcal U_M}(\widehat{\bm{\phi}}^{\,e}_j)\bigr\|_2^2.
\end{equation}
Specifically, we use the logarithmic score $\zeta_j \coloneqq \log(\mathcal R_j+\varepsilon)$ for some small $\varepsilon>0$, with smaller values indicating stronger signal-subspace consistency. Algorithm~\ref{alg:residual_selection_energy} summarizes this procedure.

Section~\ref{sec:numerical_experiments} evaluates this score as a mode-selection method. Here, the point is the residual geometry:
If $\widehat{\bm{\phi}}^{\,e}_j$ is a true mode, it is dominated by the signal
component, so its energy concentrates on the signal subspace $\mathcal S$ and
therefore, by Assumption~\ref{ass:subspace_distance_bound}, on $\mathcal U_m$.
Since $\mathcal U_m\subset\mathcal U_M$, a true mode typically places very
little energy outside $\mathcal U_M$ (and hence in $\mathcal U_M^\perp$), so
$\bigl\|\bm r_{\mathcal U_M}(\widehat{\bm{\phi}}^{\,e}_j)\bigr\|_2^2$ is small.

Conversely, if $\widehat{\bm{\phi}}^{\,e}_j$ is a spurious mode, it is a noise-driven artifact of the finite-rank least-squares
fit rather than of coherent low-dimensional structure in the data. Such modes
are not associated with the signal subspace and therefore are not confined to
any low-dimensional subspace, typically placing non-negligible energy outside
$\mathcal U_M$. This behavior is consistent with well-documented observations
of spectral pollution and noise-induced modes in practical \ac{DMD} computations
\cite{Schmid2022,ColbrookAytonSzoke2023}. Consequently, a spurious mode generally
has a larger component in $\mathcal U_M^\perp$, and thus
$\|\bm r_{\mathcal U_M}(\widehat{\bm{\phi}}^{\,e}_j)\|_2^2$ is large relative to
that of a true mode.
Section~\ref{sec:kv_structure} further develops the geometric viewpoint, and its significance, by showing that it governs the internal structure of the modes.

\begin{algorithm}
  \centering
  \setlength{\fboxsep}{6pt}
  \setlength{\fboxrule}{0.4pt}
  \fbox{%
    \begin{minipage}{0.95\linewidth}
      \begin{algorithmic}[1]
        \Require Eigenpairs $\{(\widehat{\lambda}_j,\widehat{\bm{\phi}}^{\,e}_j)\}_{j=1}^M$, small $\varepsilon>0$
        \Ensure Per-mode scores $\{\zeta_j\}_{j=1}^M$ (smaller is more likely true)
        \For{$j = 1, \dots, M$}
          \State $\mathcal{R}_j \gets \|\widehat{\bm{\phi}}^{\,e}_j\|_2^2 - |\widehat{\lambda}_j|^2$ \Comment{by \eqref{eq:esr_energy_from_mode_and_eig}}
          \State $\zeta_j \gets \log(\mathcal{R}_j + \varepsilon)$
        \EndFor
      \end{algorithmic}
    \end{minipage}%
  }
  \caption{\ac{ESR}-energy score per mode.}
  \label{alg:residual_selection_energy}
\end{algorithm}

To compute the residual norm efficiently, following~\eqref{eq:proj_identity}, we decompose the exact mode as
\begin{equation}\label{eq:exact_mode_decomp}
\widehat{\bm{\phi}}^{\,e}_j
=
\widehat{\lambda}_j\,\widehat{\bm{\phi}}^{\,p}_j
+
\bm r_{\mathcal{U}_M}(\widehat{\bm{\phi}}^{\,e}_j).
\end{equation}

Since $\|\widehat{\bm{\phi}}^{\,p}_j\|_2=1$ and $\bm r_{\mathcal{U}_M}(\widehat{\bm{\phi}}^{\,e}_j)\in \mathcal U_M^\perp$, the two terms in~\eqref{eq:exact_mode_decomp} are orthogonal, and hence
\begin{equation}\label{eq:exact_mode_energy_decomp}
\bigl\|\widehat{\bm{\phi}}^{\,e}_j\bigr\|_2^2
=
|\widehat{\lambda}_j|^2
+
\bigl\|\bm r_{\mathcal{U}_M}(\widehat{\bm{\phi}}^{\,e}_j)\bigr\|_2^2.
\end{equation}
Therefore,
\begin{equation}\label{eq:esr_energy_from_mode_and_eig}
\bigl\|\bm r_{\mathcal{U}_M}(\widehat{\bm{\phi}}^{\,e}_j)\bigr\|_2^2
=
\|\widehat{\bm{\phi}}^{\,e}_j\|_2^2 - |\widehat{\lambda}_j|^2 .
\end{equation}
The expression in~\eqref{eq:esr_energy_from_mode_and_eig} implies that the residual norms can be computed directly from the norms of the exact modes and the magnitudes of the eigenvalues, without forming $\bm P_{\mathcal U_M}$. This adds only a negligible additional cost beyond the cost of \ac{SVD} used in the standard delay-coordinates \ac{DMD}. For a complete complexity analysis, see Appendix~\ref{app:complexity}.

Since the scores $\mathcal R_j$ can be computed directly from the rank-$M$ \ac{DMD} eigenpairs via \eqref{eq:esr_energy_from_mode_and_eig}, the ESR-energy score applies to the classical case, where $L=1$ without delay-coordinates.

%% file: figures/residual_visualization.tex
\begin{figure*}[t]
  \centering
  \begin{tikzpicture}
  \tikzset{>=stealth}
    \begin{groupplot}[
      group style={
        group size=2 by 1,
        horizontal sep=2cm,
      },
      view={45}{20},
      width=0.48\textwidth,
      height=0.4\textwidth,
      xmin=0, xmax=1.6,
      ymin=0, ymax=1.6,
      zmin=0, zmax=0.75,
      axis lines=none,
      ticks=none,
      clip=false,
    ]

      \nextgroupplot[title={(a) Signal Subspace Residual (SSR)}]

        \fill[gray!2] (axis cs:0,0,0) -- (axis cs:1.5,0,0) -- (axis cs:1.5,1.5,0) -- (axis cs:0,1.5,0) -- cycle;
        \pgfplotsinvokeforeach{0,0.25,...,1.5}{
            \draw[gray!30, thin] (axis cs:#1,0,0) -- (axis cs:#1,1.5,0);
            \draw[gray!30, thin] (axis cs:0,#1,0) -- (axis cs:1.5,#1,0);
        }

        \addplot3[->, black, line width=0.8pt] coordinates {(0,0,0) (1.6,0,0)}; \node[anchor=west] at (axis cs:1.6,0,0) {$u_1$};
        \addplot3[->, black, line width=0.8pt] coordinates {(0,0,0) (0,1.6,0)}; \node[anchor=south] at (axis cs:0,1.6,0) {$u_2$};
        \addplot3[->, black, line width=0.8pt] coordinates {(0,0,0) (0,0,0.75)}; \node[anchor=south] at (axis cs:0,0,0.73) {$u_3$};
        \fill[black] (axis cs:0,0,0) circle (1pt);

        \addplot3[->, dotted, line width=1pt, draw=green!60!black] coordinates {(0,0,0) (1.5, 0.23, 0)};
        \node[anchor=west, text=green!60!black] at (axis cs:1.5, 0.23, 0) {$\mathcal{S}$};

        \draw[dotted, black, line width=0.8pt] (axis cs:1.44, 0, 0) to [bend right=20] (axis cs:1.4, 0.2, 0);
        \node[anchor=south west, font=\scriptsize, inner sep=1pt] at (axis cs: 1.45, -0.5, 0.0) {$\arcsin(\eta)$};

        \addplot3[->, line width=1pt, draw=blue] coordinates {(0,0,0) (1, 0.25, 0.1)};
        \node[anchor=south east, text=blue] at (axis cs:1.3, 0.175, 0.13) {$\bm{\widehat{\phi}^{e}}_{ \text{true} }$};

        \addplot3[->, line width=1pt, draw=red] coordinates {(0,0,0) (0.1025660859, 1.025660859, 0.65)};
        \node[anchor=south east, text=red] at (axis cs:0.12, 0.66, 0.45) {$\bm{\widehat{\phi}^{e}}_{ \text{spur} }$};

        \addplot3[->, line width=1pt, draw=violet] coordinates {(1.0143, 0.1555, 0) (1, 0.25, 0.1)};
        \node[anchor=west, text=violet, font=\footnotesize, inner sep=2pt] at (axis cs:1, 0.25, 0.05) {True-SSR};

        \addplot3[->, line width=1pt, draw=violet] coordinates {(0.2538, 0.0389, 0) (0.1025660859, 1.025660859, 0.65)};
        \node[anchor=west, text=violet, font=\footnotesize, inner sep=2pt] at (axis cs:0.2, 0.6, 0.4) {Spurious-SSR};

        \draw[black, line width=0.4pt]
          (axis cs: 1.0536, 0.1616, 0) --
          (axis cs: 1.0536, 0.1889, 0.0289) --
          (axis cs: 1.0143, 0.1829, 0.0289);

        \draw[black, line width=0.4pt]
          (axis cs: 0.2931, 0.0450, 0) --
          (axis cs: 0.2886, 0.0843, 0.0285) --
          (axis cs: 0.2493, 0.0782, 0.0285);

      \nextgroupplot[title={(b) Estimated Subspace Residual (ESR)}]

        \fill[gray!2] (axis cs:0,0,0) -- (axis cs:1.5,0,0) -- (axis cs:1.5,1.5,0) -- (axis cs:0,1.5,0) -- cycle;
        \pgfplotsinvokeforeach{0,0.25,...,1.5}{
            \draw[gray!30, thin] (axis cs:#1,0,0) -- (axis cs:#1,1.5,0);
            \draw[gray!30, thin] (axis cs:0,#1,0) -- (axis cs:1.5,#1,0);
        }

        \addplot3[->, black, line width=0.8pt] coordinates {(0,0,0) (1.6,0,0)}; \node[anchor=west] at (axis cs:1.6,0,0) {$u_1$};
        \addplot3[->, black, line width=0.8pt] coordinates {(0,0,0) (0,1.6,0)}; \node[anchor=south] at (axis cs:0,1.6,0) {$u_2$};
        \addplot3[->, black, line width=0.8pt] coordinates {(0,0,0) (0,0,0.75)}; \node[anchor=south] at (axis cs:0,0,0.73) {$u_3$};
        \fill[black] (axis cs:0,0,0) circle (1pt);

        \addplot3[->, dotted, line width=1pt, draw=green!60!black] coordinates {(0,0,0) (1.5, 0.23, 0)};
        \node[anchor=west, text=green!60!black] at (axis cs:1.5, 0.23, 0) {$\mathcal{S}$};

        \draw[dotted, black, line width=0.8pt] (axis cs:1.44, 0, 0) to [bend right=20] (axis cs:1.4, 0.2, 0);
        \node[anchor=south west, font=\scriptsize, inner sep=1pt] at (axis cs: 1.45, -0.5, 0.0) {$\arcsin(\eta)$};

        \addplot3[->, line width=1pt, draw=blue] coordinates {(0,0,0) (1, 0.25, 0.1)};
        \node[anchor=south east, text=blue] at (axis cs:1.3, 0.175, 0.13) {$\bm{\widehat{\phi}^{e}}_{ \text{true} }$};

        \addplot3[->, line width=1pt, draw=red] coordinates {(0,0,0) (0.1025660859, 1.025660859, 0.65)};
        \node[anchor=south east, text=red] at (axis cs:0.12, 0.66, 0.45) {$\bm{\widehat{\phi}^{e}}_{ \text{spur} }$};

        \addplot3[->, dashed, line width=0.8pt, draw=blue!80] coordinates {(0,0,0) (1, 0.25, 0)};
        \addplot3[->, line width=1pt, draw=orange] coordinates {(1, 0.25, 0) (1, 0.25, 0.1)};
        \node[anchor=west, text=orange, font=\footnotesize] at (axis cs:1, 0.25, 0.05) {True-ESR};

        \addplot3[->, dashed, line width=0.8pt, draw=red!80] coordinates {(0,0,0) (0.1025660859, 1.025660859, 0)};
        \addplot3[-, line width=1pt, draw=orange] coordinates {(0.1025660859, 1.025660859, 0) (0.1025660859, 1.025660859, 0.65)};
        \node[anchor=west, text=orange, font=\footnotesize] at (axis cs:0.1, 1, 0.4) {Spurious-ESR};

    \end{groupplot}
  \end{tikzpicture}
  \caption{Geometric illustration of signal-subspace residual (SSR) and estimated-subspace residual (ESR) in the case \(m=1\) and \(M=2\), where the true signal subspace \(\mathcal{S}\) (green) lies inside the truncation subspace \(\mathcal{U}_M = \operatorname{span}\{u_1,u_2\}\). True and spurious modes are drawn in blue and red, respectively, together with their projections onto \(\mathcal{U}_M\) (dashed). (a) SSR corresponds to the in-plane deviation of a mode from \(\mathcal{S}\), shown as violet arrows. (b) ESR corresponds to the out-of-plane component of a mode, shown as orange arrows. In this constructed example, the true mode has small SSR and ESR, whereas the spurious mode has larger residual in both senses.}
  \label{fig:two_panel_residual}
\end{figure*}

%% file: sections/kv_structure_and_nested_DMD.tex
\section{Kronecker--Vandermonde structure in modes}\label{sec:kv_structure}

Assuming an exponential signal model as in~\eqref{eq:explicit_signal_structure},
Bronstein \emph{et al.}\cite{Bronstein2022} showed that, when using delay-coordinates, each true mode
exhibits a \ac{KV} structure matching the lifted underlying mode $\widetilde{\bm{\phi}}_j$ of~\eqref{eq:lifted_signal_structure},
\begin{equation}\label{eq:kv_kronecker_def}
\bm v_L(\lambda_j)\otimes\bm{\phi}_j
\;=\;
\begin{bmatrix}
\bm{\phi}_j \\
\lambda_j\,\bm{\phi}_j \\
\vdots \\
\lambda_j^{L-1}\,\bm{\phi}_j
\end{bmatrix}
\in \mathbb C^{DL},
\end{equation}
where $\bm v_L(\lambda)\in\mathbb C^{L}$ is the Vandermonde vector:
\begin{equation}\label{eq:vandermonde_vector_def}
\bm v_L(\lambda)\;=\;[\,1,\ \lambda,\ \ldots,\ \lambda^{L-1}\,]^\top,
\end{equation}
and $\otimes$ is the Kronecker product, defined by
\begin{equation}\label{eq:kronecker_def}
\bm a\otimes\bm b
\;\coloneqq\;
\begin{bmatrix}
a_1\,\bm b\\
a_2\,\bm b\\
\vdots\\
a_L\,\bm b
\end{bmatrix}
\in\mathbb C^{DL}.
\end{equation}
for any $\bm a\in\mathbb C^{L}$ and $\bm b\in\mathbb C^{D}$.
See App.~\ref{app:data_side_kv} for a compact derivation of this result.
Here, we extend this analysis in three ways. First, we derive the same \ac{KV} structure without assuming a signal model, showing that this structure is induced by the delay-coordinates rather than by the signal model. Second, our new derivation shows that KV structure can arise in both \emph{true}
and \emph{spurious} components, often with different strengths (with true modes adhering significantly more strongly). This motivates continuous KV-deviation metrics that quantify the extent to which each mode conforms to the KV template, while avoiding undue sensitivity to small, localized perturbations.
Third, we link the new derivation of the KV structure to the \ac{ESR} (Definition~\ref{def:esr_vector})
and use their relationship to motivate the KV-deviation scores developed below.

We note that this \ac{KV} structure implies that, in delay-coordinates, a mode no longer represents a purely spatial pattern in $\mathbb{C}^{D}$, but rather a coupled spatiotemporal pattern in $\mathbb{C}^{DL}$: its $D$-dimensional lagged segments are tied across lags by a geometric progression parameterized by $\lambda_j$. This structure is illustrated schematically in Fig.~\ref{fig:kv_structure}.

\subsection{Block-companion minimizer and compression}
\label{subsec:block_companion}

To derive the KV structure without assuming a particular signal model, we examine the following unconstrained least-squares problem:
\begin{equation}\label{eq:least_squares_objective}
\min_{\bm A\in\mathbb C^{DL\times DL}} \|\bm X_1-\bm A\bm X_0\|_F^2,
\end{equation}
where $\bm X_0,\bm X_1\in\mathbb C^{DL\times N}$ are the snapshot matrices defined in~\eqref{eq:x0_x1}.
Although the minimizer of~\eqref{eq:least_squares_objective} is not necessarily unique, there exists a minimizer whose structure is dictated by the delay-coordinates.
\begin{lemma}[Block-companion minimizer]\label{lem:companion_minimizer}
The optimization problem~\eqref{eq:least_squares_objective} admits the minimizer $\bm{C}_L$ with the following block-companion form:
\begin{equation}\label{eq:companion_def}
\bm C_L
=
\begin{bmatrix}
\bm 0   & \bm I_D &        &        & \bm 0 \\
        & \ddots  & \ddots &        &       \\
        &         & \bm 0  & \bm I_D & \bm 0 \\
\bm B_1 & \bm B_2 & \cdots & \bm B_{L-1} & \bm B_L
\end{bmatrix}
\in\mathbb C^{DL\times DL},
\end{equation}
where $\bm B_\ell\in\mathbb C^{D\times D}$.
\end{lemma}
\noindent\textit{Proof}
See Appendix~\ref{app:existence_block_companion}.

Lemma~\ref{lem:companion_minimizer} shows that the delay-coordinates are
explicitly encoded in a minimizer: the top $L-1$ block rows implement the
one-step shift between consecutive delay blocks, while the last block row
defines a linear predictor of the next sample from the stacked history
$\widetilde{\bm x}_k=[\,\bm x_k^\top,\ldots,\bm x_{k+L-1}^\top\,]^\top$, namely
$\bm x_{k+L}\approx \bm B\,\widetilde{\bm x}_k$, where
$\bm B=[\,\bm B_1~\cdots~\bm B_L\,]\in\mathbb C^{D\times DL}$.
Throughout, we fix $\bm C_L$ by taking $\bm B$ to be the Moore-Penrose
(pseudoinverse) least-squares predictor; Appendix~\ref{app:existence_block_companion}
makes this choice explicit.

\begin{lemma}[Eigenvectors of block-companion matrices (see, e.g., Ref.~\onlinecite{MackeyMackeyMehlMehrmann2006})]
\label{lem:kv_eigenvectors}
Any eigenvector of $\bm C_L$ associated with eigenvalue $\mu\in\mathbb C$ has the following
Kronecker--Vandermonde form:
\begin{equation}\label{eq:kv_eigenvector_form}
\bm v_L(\mu)\otimes \bm \varphi,
\end{equation}
where $\bm v_L(\mu)=[ 1, \mu, \ldots,  \mu^{L-1}]^\top$, and
$\bm \varphi\in\mathbb C^{D}\setminus\{\bm 0\}$.
\end{lemma}

\noindent\textit{Proof.}
See Appendix~\ref{app:kv_eigenvectors}.

We now connect the block-companion solution $\bm C_L\in\mathbb C^{DL\times DL}$ of \eqref{eq:companion_def} to the truncated DMD propagator $\bm A_M\in\mathbb C^{M\times M}$ defined in \eqref{eq:reduced_propagator}.

\begin{proposition}
\label{prop:AM_is_compression_of_CL}
The truncated \ac{DMD} propagator can be recast as
\begin{equation}\label{eq:AM_is_compression_of_CL}
\bm A_M \;=\; \bm U_M^{\!H}\,\bm C_L\,\bm U_M,
\end{equation}
where $\bm U_M$ is defined in \eqref{eq:svd_x0}.
\end{proposition}
\noindent\textit{Proof.}
See Appendix~\ref{app:reduced-propagator-is-compression}.

Based on Proposition \ref{prop:AM_is_compression_of_CL}, we connect the eigenvectors of $\bm C_L$, which have a KV structure (Lemma~\ref{lem:kv_eigenvectors}), to the projected-\ac{DMD} modes.
\begin{theorem}
\label{thm:residual_companion_identity}
Let $(\widehat{\lambda}_j,\bm w_j)$ be an eigenpair of the reduced propagator $\bm A_M$
defined in~\eqref{eq:reduced_propagator}, and let
$\widehat{\bm\phi}^{\,p}_j$ and $\widehat{\bm\phi}^{\,e}_j$ be the associated
projected and exact \ac{DMD} modes defined in~\eqref{eq:proj_mode_def} and \eqref{eq:exact_mode_def}, respectively.
Then, $(\widehat\lambda_j,\widehat{\bm\phi}^{\,p}_j)$ satisfies
\begin{equation}\label{eq:main_identity_vector}
(\bm C_L-\widehat{\lambda}_j\bm I)\,\widehat{\bm\phi}^{\,p}_j
\;=\;
\bm r_{\mathcal U_M}\!\bigl(\widehat{\bm\phi}^{\,e}_j\bigr),
\end{equation}
where $\bm r_{\mathcal U_M}(\widehat{\bm\phi}^{\,e}_j)$ is the corresponding estimated-subspace residual vector (Definition~\ref{def:esr_vector}). 
\end{theorem}

\noindent\textit{Proof.}
See Appendix~\ref{app:proof_residual_companion_identity}.
Theorem~\ref{thm:residual_companion_identity} shows that the estimated subspace residual vector $\bm r_{\mathcal U_M}(\widehat{\bm\phi}^{\,e}_j)$ governs (through the spectrum of $\bm C_L$) the extent to which the projected mode $\widehat{\bm\phi}^{\,p}_j$ conforms to the KV structure exhibited by the eigenvectors of the block-companion minimizer $\bm C_L$. In particular, $\bm r_{\mathcal U_M}\!\bigl(\widehat{\bm\phi}^{\,e}_j\bigr)=\bm 0$ implies that $\widehat{\bm\phi}^{\,p}_j$ is an eigenvector of $\bm C_L$, and hence has exact Kronecker--Vandermonde structure.
Specifically, $\bm r_{\mathcal U_M}(\widehat{\bm\phi}^{\,e}_j)$ is exactly the eigenrelation residual of the projected eigenpair $(\widehat{\lambda}_j,\widehat{\bm\phi}^{\,p}_j)$ with respect to $\bm C_L$.

Theorem~\ref{thm:residual_companion_identity} therefore provides a structural link motivating KV-deviation metrics: the \ac{ESR} norm $\bigl\|\bm r_{\mathcal U_M}(\widehat{\bm\phi}^{\,e}_j)\bigr\|_2$ directly measures the extent to which the projected \ac{DMD} mode violates the block-companion eigenrelation. Since the eigenvectors of $\bm C_L$ have KV structure, smaller \ac{ESR} suggests closer KV consistency, while larger \ac{ESR} indicates a stronger departure from the ideal KV form. In Sec.~\ref{sec:mode_geometry}, the \ac{ESR} norm was used as a per-mode score for ranking true and spurious components; the link established here motivates KV deviation as a score that targets the structural property directly.

\subsection{Nested DMD}
\label{subsec:kv_nested_dmd}
Above, we showed that delay-embedded \ac{DMD} induces \ac{KV} structure in the computed modes. We postulate that true modes adhere more closely to the KV template than spurious modes. We therefore introduce direct KV-deviation scores computed from the mode entries and $\widehat{\lambda}_j$, and use them as per-mode structural scores.

To motivate the proposed quantification, consider the \emph{ideal KV template} for a projected \ac{DMD} mode:
\begin{equation}\label{eq:kv_mode}
\widehat{\bm\phi}^{\,p}_j \;=\; \bm v_L(\widehat{\lambda}_j)\otimes \widehat{\bm\phi}^{(0)}_j,
\qquad
\widehat{\bm\phi}^{(0)}_j \in \mathbb{C}^{D},
\end{equation}
where $\bm v_L(\lambda)\coloneqq[1,\lambda,\ldots,\lambda^{L-1}]^\top$.
Reshaping this mode, of length $DL$, as a $D\times L$ matrix by organizing blocks of $D$ coordinates as $L$ columns, yields 
\begin{equation}\label{eq:kv_mode_matrix_ideal}
\bigl[\,\widehat{\bm\phi}^{(0)}_j,\;
\widehat{\lambda}_j\,\widehat{\bm\phi}^{(0)}_j,\;
\ldots,\;
\widehat{\lambda}_j^{\,L-1}\widehat{\bm\phi}^{(0)}_j\,\bigr]
.
\end{equation}
Viewing this matrix as a sequence of columns, the $\ell$th column can be recast as
\begin{equation}\label{eq:kv_column_exponential_form}
\check b_j\,\check{\bm\phi}_j\,\widehat{\lambda}_j^{\,\ell},
\end{equation}
where $\check b_j\coloneqq\bigl\|\widehat{\bm\phi}^{(0)}_j\bigr\|_2$ and $\check{\bm\phi}_j\coloneqq
\widehat{\bm\phi}^{(0)}_j/\bigl\|\widehat{\bm\phi}^{(0)}_j\bigr\|_2$. This expression coincides with the one-mode specialization of \eqref{eq:explicit_signal_structure},
with the lag index $\ell$ replacing the time index $k$. Therefore, this reshaped form is consistent with the signal model \emph{assumed} by an order-1 \ac{DMD}.

Consequently, we reshape each computed projected mode
$\widehat{\bm\phi}^{\,p}_j=[\,\widehat{\bm\phi}^{(0)}_j;\ldots;\widehat{\bm\phi}^{(L-1)}_j\,]$ as above, obtaining the following mode matrix,
\begin{equation}\label{eq:mode_matrix}
\widehat{\bm \Phi}_j
\;\coloneqq\;
\bigl[\,\widehat{\bm\phi}^{(0)}_j,\ldots,\widehat{\bm\phi}^{(L-1)}_j\,\bigr]
\in \mathbb{C}^{D\times L}.
\end{equation}
Then, we apply an order-1 \ac{DMD} to the column sequence of $\widehat{\bm \Phi}_j$, along the lag axis, yielding a single \ac{DMD} eigenvalue $\check{\lambda}_{j}\in\mathbb C$
and a corresponding \ac{DMD} mode
$\check{\bm\phi}^{(0)}_{j}\in\mathbb C^{D}$.

Following the standard \ac{DMD} reconstruction convention, we reconstruct the rank-1 mode matrix by:
\begin{equation}\label{eq:kv_reconstruction}
\widehat{\bm \Phi}^{(\mathrm{KV})}_j
\;\coloneqq\;
\bigl(\check{\bm\phi}^{(0)}_j\,\check{\bm\phi}^{(0)H}_j\,\widehat{\bm\phi}^{(0)}_j\bigr)\,
\bm v_L^{\!\top}(\check{\lambda}_j),
\end{equation}
where $\widehat{\bm\phi}^{(0)}_j$ is the first column of $\widehat{\bm \Phi}_j$.

Based on the reconstructed matrix $\widehat{\bm \Phi}^{(\mathrm{KV})}_j$, we define a KV-deviation score as the mean squared reconstruction error:
\begin{equation}\label{eq:kv_scores_def}
\mathcal R^{(\mathrm{KV})}_j
\;\coloneqq\;
\frac{\bigl\|\widehat{\bm \Phi}_j-\widehat{\bm \Phi}^{(\mathrm{KV})}_j\bigr\|_F^2}{DL}.
\end{equation}

The score $\mathcal R^{(\mathrm{KV})}_j$ is a direct KV-deviation metric: it
vanishes when the reshaped mode follows an exact rank-1 KV template along the
lag axis and increases as the mode departs from KV structure.

From this point onward we follow Algorithm~\ref{alg:residual_selection_energy}: we apply the same score scaling. Here the per-mode feature is scalar and derived from $\mathcal R^{(\mathrm{KV})}_j$. Algorithm~\ref{alg:kv_nested_rank1} summarizes this procedure.

\begin{algorithm}
  \centering
  \setlength{\fboxsep}{6pt}
  \setlength{\fboxrule}{0.4pt}
  \fbox{%
    \begin{minipage}{0.95\linewidth}
      \begin{algorithmic}[1]
        \Require Projected modes $\{\widehat{\bm\phi}^{\,p}_j\}_{j=1}^M \subset \mathbb{C}^{DL}$, embedding length $L$, spatial dimension $D$, small $\varepsilon>0$
        \Ensure Per-mode scores $\{\zeta_j\}_{j=1}^M$ (smaller is more likely true)
        \For{$j = 1, \dots, M$}
          \State $\widehat{\bm \Phi}_j \gets \mathrm{reshape}(\widehat{\bm\phi}^{\,p}_j,\, D\times L)$
          \State Run rank-1 DMD on $\widehat{\bm \Phi}_j$ to obtain $(\check{\bm\phi}^{(0)}_j,\check{\lambda}_j)$
          \State $\bm y_0 \gets$ first column of $\widehat{\bm \Phi}_j$
          \State $\widehat{\bm \Phi}^{(\mathrm{KV})}_j \gets \bigl(\check{\bm\phi}^{(0)}_j\,\check{\bm\phi}^{(0)H}_j\,\bm y_0\bigr)\,
\bm v_L(\check{\lambda}_j)^{\!\top}$
          \State $\mathcal R^{(\mathrm{KV})}_j \gets \|\widehat{\bm \Phi}_j - \widehat{\bm \Phi}^{(\mathrm{KV})}_j\|_F^2 /(DL)$
          \State $\zeta_j \gets \log(\mathcal R^{(\mathrm{KV})}_j+\varepsilon)$
        \EndFor
      \end{algorithmic}
    \end{minipage}%
  }
  \caption{\ac{KV}-deviation score via nested rank-1 DMD.}
  \label{alg:kv_nested_rank1}
\end{algorithm}

\begin{figure*}[t]
  \centering
  \includegraphics[width=0.75\textwidth]{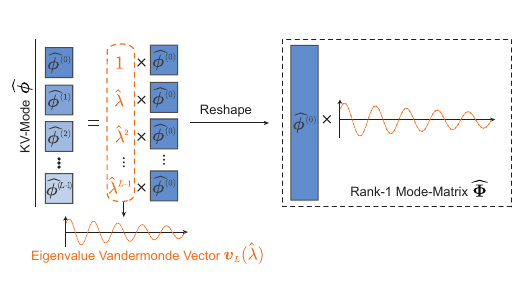}
  \caption{
  Kronecker--Vandermonde structure of a \ac{DMD} mode.
  A KV mode is formed by stacking $L$ lag segments of length $D$, where each
  segment is a scaled and rotated copy of a base spatial vector
  $\widehat{\bm{\phi}}^{(0)}$. The multipliers $\widehat{\lambda}^{\ell}$ form the
  Vandermonde vector $\bm v_L(\widehat{\lambda})=[1,\widehat{\lambda},\ldots,\widehat{\lambda}^{L-1}]^\top$,
  illustrated at the bottom. Reshaping the KV mode yields the rank-1 outer
  product $\widehat{\bm{\phi}}^{(0)}\,\bm v_L(\widehat{\lambda})^\top$, mirroring the
  rank-1 structure underlying order-1 \ac{DMD}.
}

  \label{fig:kv_structure}
\end{figure*}

The nested rank-1 construction replaces one delay-coordinates \ac{DMD} fit of
unknown order with $M$ independent order-1 fits along the lag axis, one per
mode.

The added cost scales linearly with $D$, $L$, and $M$ and is negligible compared
with the truncated \ac{SVD} used to form the projected modes; see
Appendix~\ref{app:complexity} for detailed asymptotic costs.

As a simpler variant, we consider a \emph{fixed-eigenvalue KV fit} (FEKVF).
Here ``fixed eigenvalue'' means that we do not re-estimate a lag-axis
eigenvalue: we set the lag multiplier to the original \ac{DMD} eigenvalue
$\widehat{\lambda}_j$ and fit only the spatial coefficient in the KV template. For
each mode this amounts to the least-squares problem
\[
\min_{\bm \varphi\in\mathbb C^{D}}
\bigl\|\widehat{\bm \Phi}_j-\bm \varphi\,\bm v_L(\widehat{\lambda}_j)^{\!\top}\bigr\|_F^2,
\]
yielding a single scalar deviation score in closed form
(Appendix~\ref{app:fekvf}). 
Compared with the nested rank-1 method, FEKVF
removes the inner order-1 \ac{DMD} step and is therefore computationally cheaper per mode.

%% file: sections/connection_to_existing_methods.tex
\section{Connection to existing mode-selection methods}
\label{sec:connection_to_existing_methods}
This section interprets existing mode-selection heuristics through the geometric framework developed above. In particular, we show that commonly used magnitude-based ranking rules implicitly conflate eigenvalue scaling with subspace residual geometry, and therefore degrade under delay-coordinates.

\subsection{Spatiotemporal coupling ratio-based method}
\label{subsec:stc_ratio_connection}
Both the approach of Bronstein \emph{et al.}\cite{Bronstein2022} and our approach use \ac{KV} structure to form per-mode
scores that are small when the structure is present. Bronstein \emph{et al.}\cite{Bronstein2022} do so
via entrywise quotient checks across consecutive lag blocks, testing whether
$(\widehat{\bm\phi}^{(\ell+1)}_j)_d/(\widehat{\bm\phi}^{(\ell)}_j)_d\approx\widehat\lambda_j$ for entries $d$ within each block,
and aggregating the results. In delay-coordinates, however, ratio tests can be
numerically fragile: whenever $(\widehat{\bm\phi}^{(\ell)}_j)_d$ is small, the
quotient is ill-conditioned and can take arbitrarily large values, so the
aggregate can be dominated by a few unstable coordinates. Our analysis suggests
that KV adherence is a continuous property shared by true and spurious modes to
different degrees, so quotient checks often behave like a near-binary test (pass or fail). In contrast, our \ac{KV}-fit score measures a graded distance from the \ac{KV} template through a global reconstruction error. We compare these criteria numerically in
Sec.~\ref{sec:numerical_experiments}.

\subsection{Mode-norm}
\label{subsec:mode_norm_analysis}
A common order-detection method in standard \ac{DMD} ranks modes by the Euclidean norm of the exact-\ac{DMD} mode, $\|\widehat{\bm{\phi}}^{\,e}_j\|_2$, motivated by the intuition that true modes carry larger energy than spurious modes. 
Our analysis in Sec.~\ref{sec:mode_geometry} sheds new light on this heuristic.
In \eqref{eq:exact_mode_energy_decomp}, the squared-norm of exact \ac{DMD} modes $\|\widehat{\bm{\phi}}^{\,e}_j\|_2^2$ is expressed as the sum of the squared magnitude of the eigenvalue $|\widehat{\lambda}_j|^2$ and the \ac{ESR} squared-norm
$\|\bm r_{\mathcal U_M}(\widehat{\bm{\phi}}^{\,e}_j)\|_2^2$.

When using delay-coordinates, we empirically observe that the magnitude of the eigenvalue loses its discriminative power.
Specifically, as the embedding length $L$ grows, the bulk of spurious eigenvalue
magnitudes concentrates in an annular band near the unit circle, as illustrated
in Fig.~\ref{fig:spur_eigs_L} (left panel). 
Intuitively, delay embedding enlarges the lifted state dimension while preserving the shift structure of the block-companion operator, increasing the effective polynomial degree governing spurious components. This induces a magnitude concentration near the unit circle, reducing the discriminative power of eigenvalue magnitude alone.
This behavior is qualitatively consistent
with annulus concentration phenomena studied for random polynomial roots; see
Soundararajan \cite{Soundararajan2019} for a detailed analysis. 

We first focus on
the case where $L$ is sufficiently large, so that spurious modes satisfy $|\widehat{\lambda}_j|\approx 1$ for  $j\in\mathcal{J}_{\mathrm{spur}}$.

In the weakly damped case, the underlying eigenvalues satisfy $|\widehat{\lambda_j}|\approx 1$, for $j\in[m]$, where we assume that the
associated \ac{DMD} eigenvalues satisfy $|\widehat{\lambda}_j|\approx|\lambda_j|$ for $j\in\mathcal{J}_{\mathrm{true}}$ to simplify the discussion.
In this case, both true and spurious eigenvalues cluster near the unit circle as $L$ increases. Consequently, eigenvalue magnitude becomes nearly independent of mode identity, and, by the decomposition of the exact \ac{DMD} mode squared norm in \eqref{eq:exact_mode_energy_decomp}, the only discriminative component of the exact-mode norm is the \ac{ESR} squared norm $|\bm r_{\mathcal U_M}(\widehat{\bm{\phi}}^{\,e}_j)|_2^2$.
Using the mode norm $\|\widehat{\bm{\phi}}^{\,e}_j\|_2^2$ therefore confounds this informative \ac{ESR} norm with the non-informative eigenvalue magnitude $|\widehat{\lambda}_j|^2$. Since the eigenvalue magnitudes may vary independently of mode identity, they mask or dilute the separation between true and spurious modes.
For this reason, it is preferable to work directly with the \ac{ESR} norm, which isolates the meaningful geometric information without contamination from eigenvalue scaling.

Next, consider a strongly damped signal, where the underlying eigenvalues satisfy
$|\lambda_j|<1$ for $j\in[m]$.
When $L$ is sufficiently large, spurious modes typically
have $|\widehat{\lambda}_j|\approx 1$, while true \ac{DMD} modes have smaller
$|\widehat{\lambda}_j|$. In addition, true modes typically have smaller \ac{ESR} energies than spurious modes. Together, these two facts imply that true modes tend to have
\emph{smaller} exact-mode norms than spurious modes. This is the opposite of how the heuristic is commonly applied, which treats larger norms as more likely to be true.

For completeness, we note that for small-$L$ embeddings (and in the standard \ac{DMD} case
$L=1$), spurious eigenvalues are often observed well inside the unit disk. In
this regime, true modes typically have significantly larger $|\widehat{\lambda}_j|^2$
than spurious modes, so mode-norm ranking \emph{is} an effective separator. However,
the \ac{ESR} norm contributes in the opposite
direction, being larger for spurious modes, and therefore it weakens the
separation provided by $|\widehat{\lambda}_j|^2$ alone.

These observations illustrate that eigenvalue magnitude and mode norm are not intrinsic indicators of dynamical relevance under delay embedding, whereas the \ac{ESR} isolates the geometric quantity that directly reflects signal-subspace consistency and mode identity.

\subsection{ResDMD residual, \ac{ESR}, and the block-companion eigenresidual}
\label{subsec:resdmd_esr_companion}

ResDMD~\cite{ColbrookTownsend2024} is framed in the Koopman operator setting, where
a (possibly nonlinear) discrete-time dynamical system is represented by a
linear infinite-dimensional operator $\mathcal K$ acting on observables
(functions of the state)~\cite{Colbrook2023,Mezic2005,Mezic2013,Budisic2012}.
{Given a candidate Koopman eigenpair $(\lambda, \bm g)$, with a coordinate vector $\bm c$ representing
$\bm g$ in a chosen finite basis, ResDMD computes a residual that
measures how well the eigenrelation $\mathcal K \bm g \approx \lambda \bm g$
is satisfied on the data. When applying ResDMD to the delay-coordinate setting (see details in Appendix~\ref{app:resdmd_in_our_setting}), the coordinate vector $\bm c$
corresponds to a left eigenvector $\bm q \in \mathbb{C}^{M}$ of the reduced propagator
$\bm A_M$ associated with the \ac{DMD} eigenvalue $\widehat\lambda$ (that is, $\bm q^{\mathrm H}\bm A_M=\widehat\lambda\,\bm q^{\mathrm H}$), namely $\bm c = \bm q$. \textcolor{blue}{As an eigenvector of $\bm K=\bm A_M^{\mathrm H}$, this candidate's Koopman eigenvalue is the conjugate $\widehat\lambda^{*}$, and it is this conjugate that enters the residual, while the candidate is paired with the \ac{DMD} mode of eigenvalue $\widehat\lambda$.} The squared ResDMD residual is then given by:
\begin{equation}
\label{eq:resdmd_residual_in_dcdmd}
\mathrm{res}^{2}(\widehat \lambda,\bm q)
=\frac{\bigl\|\bm X_1^{\mathrm H}\bm U_M \bm q-\textcolor{blue}{\widehat \lambda^{*}}\,\bm X_0^{\mathrm H}\bm U_M \bm q\bigr\|_2^{2}}
     {\bigl\|\bm X_0^{\mathrm H}\bm U_M \bm q\bigr\|_2^{2}}.
\end{equation}}

{The \ac{ESR} of Sec.~\ref{sec:mode_geometry} is also an eigenrelation
residual: by Theorem~\ref{thm:residual_companion_identity} \eqref{eq:main_identity_vector}, it equals the
eigenresidual of the projected pair $(\widehat{\lambda}_j,
\widehat{\bm\phi}^{\,p}_j)$ against the finite-dimension block-companion operator
$\bm C_L$ of Lemma~\ref{lem:companion_minimizer}.}

{Two differences are essential. (i) The reference operators differ
fundamentally: ResDMD residual approximates an eigen-residual with respect to an
infinite-dimensional Koopman operator. In comparison, \ac{ESR} is a projection residual with respect to the estimated subspace. (ii) The two act on different objects, as can be seen by recasting each as a projection residual:
\begin{equation}
\label{eq:esr_resdmd_sidebyside}
\begin{aligned}
\text{\ac{ESR} squared-norm:}\quad & \bigl\|(\bm I-\bm U_M\bm U_M^{\mathrm H})\,\bm X_1\bm V_M\bm\Sigma_M^{-1}\bm w_j \bigr\|^2, \\
\text{ResDMD numerator:}\quad & \bigl\|(\bm I-\bm V_M\bm V_M^{\mathrm H})\,\bm X_1^{\mathrm H}\bm U_M\bm q_j \bigr\|^2,
\end{aligned}
\end{equation}
where $\bm w_j$ and $\bm q_j$ are the right and left eigenvectors of $\bm A_M$. The squared-norms of both the ESR and the ResDMD numerator are orthogonal projections of an $\bm X_1$-mapped eigenvector, but onto complementary objects: the \ac{ESR} removes the $\mathrm{col}(\bm U_M)$ component of the right exact mode, while the ResDMD residual removes the $\mathrm{col}(\bm V_M)$ component of the lifted left eigenvector.}

\textcolor{blue}{We emphasize that the \ac{ESR} is therefore not a ResDMD residual, and does not inherit its Koopman spectral guarantees (see Appendix~\ref{app:resdmd_in_our_setting}).}

{For more details, including the explicit ResDMD residual, its instantiation to the delay-coordinate setting, and the implementation used as our baseline, see Appendix~\ref{app:resdmd_in_our_setting}.}

Since our methods exploit delay-coordinates induced structure, our empirical comparison with ResDMD in Sec.~\ref{sec:numerical_experiments} focuses on scenarios where delay-coordinates are strictly necessary, namely when the spatial $D$ is small.

\subsection{KROM, \ac{ESR}, and projection residuals}
\label{subsec:krom_esr}

The Koopman reduced-order modeling framework of Mohr~\emph{et al.}~\cite{MohrFonoberovaMezic2025} is a stochastic reduced-order modeling and forecasting method: it represents the dynamics as a finite Koopman mode decomposition plus a residual noise process and uses the residual statistics to attach confidence bounds to forecasts. Its order-selection heuristic, our only point of overlap, ranks the computed \ac{DMD} modes by their reconstruction amplitudes $\widehat{\bm b}$ from~\eqref{eq:reconstruction} and chooses the smallest top-amplitude subset whose reconstruction-error projection passes a Gaussianity test.

Both KROM and the \ac{ESR} are based on the same geometric primitive: an orthogonal split of a residual against a data-driven subspace;
three differences are essential.
(i) \emph{The subspace}: the \ac{ESR} uses the fixed estimated subspace $\mathcal U_M$,
whereas KROM uses a sequence of subspaces spanned by nested subsets of $m' < M$ amplitude-ordered \ac{DMD} modes.
(ii) \emph{The object}: the \ac{ESR} acts on an exact \ac{DMD} mode, whereas KROM acts on the per-snapshot reconstruction error $\widetilde{\bm x}_k - \widehat{\bm x}_k$ of~\eqref{eq:reconstruction}, under the assumption that for $m' > m$ this projected error becomes Gaussian. 
(iii) \emph{The output}: the \ac{ESR} produces per-mode scores, whereas KROM selects the smallest subset of high-amplitude modes that achieves near-Gaussianity of the projected error.
Modifying KROM to projected \ac{DMD} modes and $m' = M$, the orthogonal component of the reconstruction error reduces to $\bm r_{\mathcal U_M}(\widetilde{\bm x}_k - \widehat{\bm x}_k) = \bm r_{\mathcal U_M}(\widetilde{\bm x}_k)$, the \ac{ESR} of the snapshot itself.

%% file: sections/numerical.tex
\section{Numerical Results}
\label{sec:numerical_experiments}

This section presents numerical experiments evaluating our proposed methods
against existing baselines. We first describe the experimental setup, then
present mode-selection results in the low-spatial-dimension regime,
illustrate the same comparison on a concrete dynamical system, and finally
report empirical observations on the magnitude of spurious eigenvalues.

\subsection{Experimental setup}
\label{subsec:experimental_setup}

We simulate signals according to~\eqref{eq:measurement_model} using the linear $m$-mode signal described in~\eqref{eq:explicit_signal_structure}, with additive Gaussian noise. We note that we observed qualitatively similar performance under other noise distributions, including bi-Gaussian, heavy-tailed Student's $t$, uniform, and heteroscedastic noise, but report only Gaussian results here for brevity.
The ground-truth spatial modes $\bm{\phi}_j$ are drawn independently and normalized to unit norm.

Across experiments, we control the following generation parameters:
the \ac{SNR}, the number of true modes $m$, the eigenvalue magnitudes (we use a common value by setting $\rho_j = \rho$ for all $j$ so that damping is controlled by the single scalar $\rho$), the minimal phase gap $\Delta\theta$, and the amplitude heterogeneity $\kappa_b$.
The minimal phase gap is defined as
\begin{equation}
\Delta\theta = \min_{j\neq k} \min\bigl(|\theta_j - \theta_k|,\, 2\pi - |\theta_j - \theta_k|\bigr).
\label{eq:delta_theta}
\end{equation}
Amplitude heterogeneity is quantified by
\begin{equation}
\kappa_b = \frac{\max_j |b_j|}{\min_j |b_j|}.
\label{eq:kappa_b}
\end{equation}

Throughout this section we focus on a regime in which the spatial dimension
$D$ is small. The motivation is structural: when $D$ is small, the snapshot
matrices carry few independent measurements per time index, so delay
coordinates must absorb most of the identification burden. In this regime the
\ac{KV} structure analyzed in Sec.~\ref{sec:kv_structure} is the dominant form
of structure available, and the comparison of scoring rules becomes informative
about whether a method exploits that structure or treats the lifted state
generically. We make this point precise across the four single-parameter
sweeps reported below. This regime is also of substantial practical interest:
it arises in applications across radar and array signal processing,
communications, neuroscience, biomedical signal analysis, and localization.
Some of these applications fit naturally within \ac{DMD}, while others are
conventionally handled by closely related Hankel- and delay-based spectral
estimators (matrix pencil, ESPRIT, Prony-type), to which the proposed scores
apply directly.

We compare three of our scores, \ac{ESR}-energy
(Algorithm~\ref{alg:residual_selection_energy}), \ac{KV} deviation via nested
rank-1 DMD (Algorithm~\ref{alg:kv_nested_rank1}, denoted NestedDMD), and the
\ac{FEKVF} variant, against two baselines: \ac{STC}~\cite{Bronstein2022}
and \textcolor{blue}{a same-data, data-adaptive, linear $\bm U_M$-dictionary ResDMD-style residual (detailed in Appendix~\ref{app:resdmd_in_our_setting}), which is a reasonable empirical baseline that does not carry all theoretical guarantees presented in}~\cite{ColbrookTownsend2024}. To compare methods
independently of any hard-decision rule, we report the
\emph{precision--recall AUC} (PR-AUC) of the induced true-vs-spurious
ranking task, evaluated per trial and averaged across trials.

\subsection{Mode classification under low spatial dimension}
\label{subsec:mode_classification_low_D}

We first evaluate mode selection for $D=2$, with two independent measurements per time index. The working point is $N=200$,
$L=66$, $M=15$, $\rho=1.0$, $\Delta\theta=0.007$, $\kappa_b=2.5$,
\ac{SNR}$=0$\,dB; in each experiment we vary one of these parameters while
holding the rest fixed. Across the four figures of this subsection, the panels show results for $m=2$ (left), $m=3$ (center), and $m=5$
(right).

Figure~\ref{fig:pr_auc_snr} shows the PR-AUC of mode selection as a
function of the \ac{SNR}. Across the tested range, the two \ac{KV}-aware scores
(NestedDMD, \ac{FEKVF}) saturate at the lowest \ac{SNR}, with
\ac{ESR}-energy close behind and ResDMD reaching the same
plateau at slightly higher \ac{SNR}. \ac{STC} is the slowest to rise and is
consistently the weakest of the methods. Increasing $m$ shifts all transitions
toward higher \ac{SNR} while preserving the relative ordering of the methods.

\begin{figure*}[t]
  \centering
  \includegraphics[width=\linewidth]{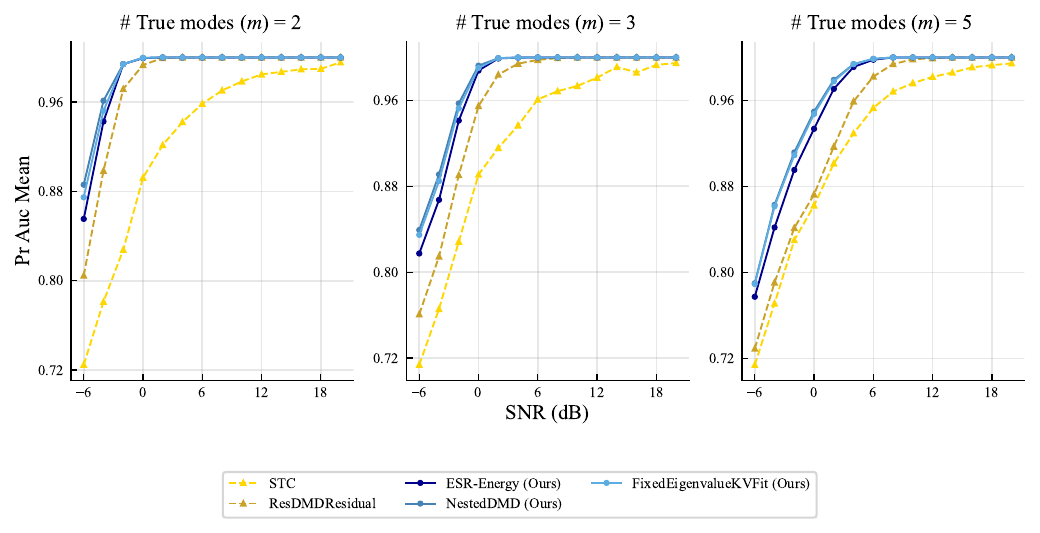}
   \caption{PR-AUC of mode classification vs.\ \ac{SNR} at $D=2$ under
   Gaussian noise. Panels show $m=2$ (left), $m=3$ (center), and $m=5$
   (right). Working point: $N=200$, $L=66$, $M=15$, $\rho=1.0$,
   $\Delta\theta=0.007$, $\kappa_b=2.5$.}
  \label{fig:pr_auc_snr}
\end{figure*}

Figure~\ref{fig:pr_auc_dtheta} reports the PR-AUC versus the minimal phase
separation $\Delta\theta$, sampled more densely around the Rayleigh limit,
where transitions are sharpest. The ordering of the methods matches
Fig.~\ref{fig:pr_auc_snr}: the \ac{KV}-aware scores saturate at the smallest
separations, \ac{ESR}-energy is close behind, and ResDMD
saturates only at larger separations. \ac{STC} plateaus near $0.9$ and
remains below the other methods within the tested range.

\begin{figure*}[t]
  \centering
  \includegraphics[width=\linewidth]{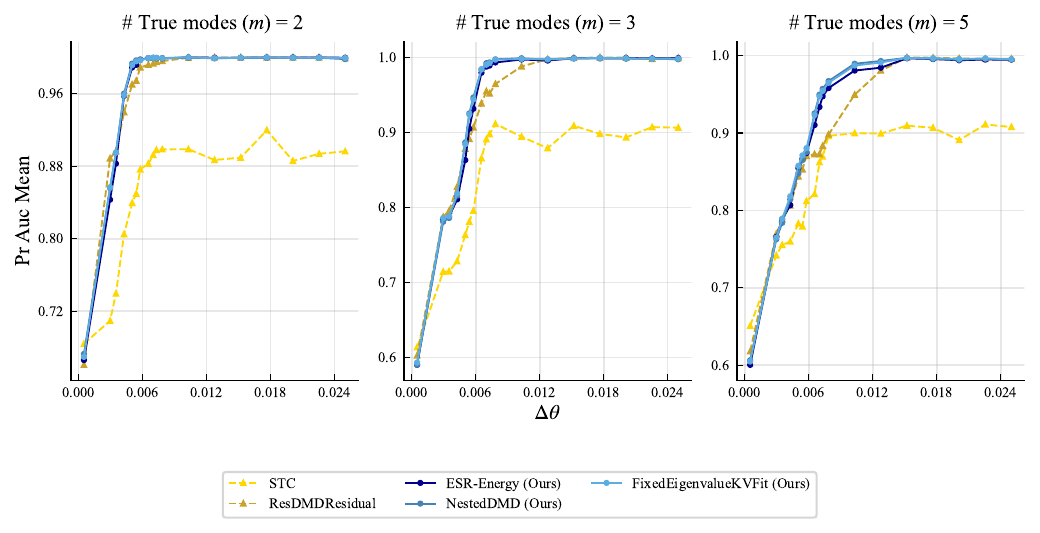}
   \caption{PR-AUC of mode classification vs.\ phase separation $\Delta\theta$
   (defined in \eqref{eq:delta_theta}) at $D=2$ under Gaussian noise.
   Panels show $m=2$ (left), $m=3$ (center), and $m=5$ (right).
   Working point matches Fig.~\ref{fig:pr_auc_snr}, with \ac{SNR} fixed at
   $0$\,dB;}
  \label{fig:pr_auc_dtheta}
\end{figure*}

Figure~\ref{fig:pr_auc_kappa} reports the PR-AUC versus the amplitude
heterogeneity $\kappa_b$. The \ac{KV}-aware scores achieve the best performance across the
tested range, followed by the \ac{ESR}-energy, ResDMD, and \ac{STC}. All methods
degrade as $\kappa_b$ grows.

\begin{figure*}[t]
  \centering
  \includegraphics[width=\linewidth]{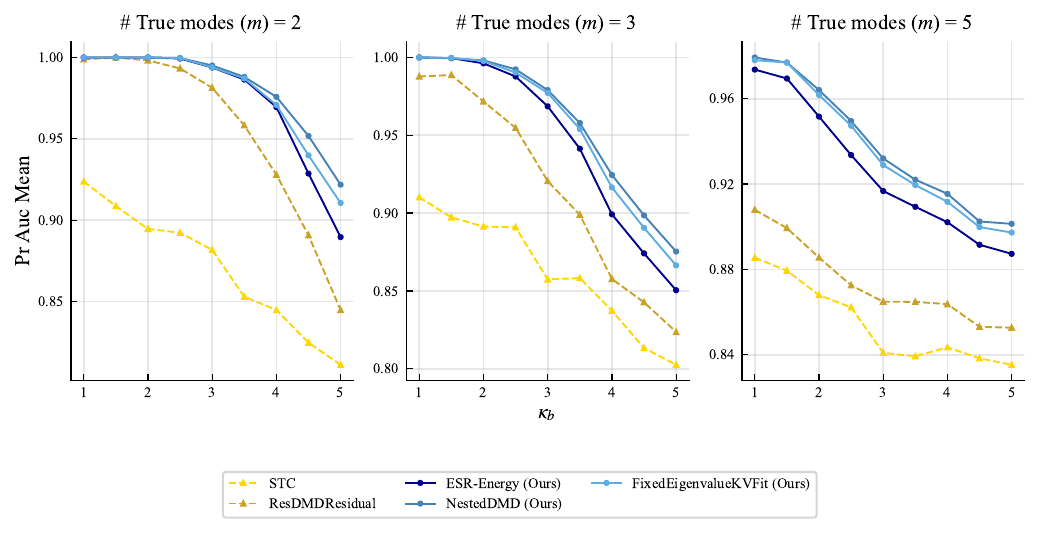}
   \caption{PR-AUC of mode classification vs.\ amplitude heterogeneity
   $\kappa_b$ (defined in \eqref{eq:kappa_b}) at $D=2$ under Gaussian noise.
   Panels show $m=2$ (left), $m=3$ (center), and $m=5$ (right).
   Working point matches Fig.~\ref{fig:pr_auc_snr};}
  \label{fig:pr_auc_kappa}
\end{figure*}

Figure~\ref{fig:pr_auc_rho} reports the PR-AUC as a function of the common
eigenvalue magnitude $\rho$. At small $\rho$, the baselines \ac{STC} and
ResDMD perform better than our scores. As $\rho \to 1$, the
\ac{KV}-aware scores and the \ac{ESR}-energy outperform the baselines.

\begin{figure*}[t]
  \centering
  \includegraphics[width=\linewidth]{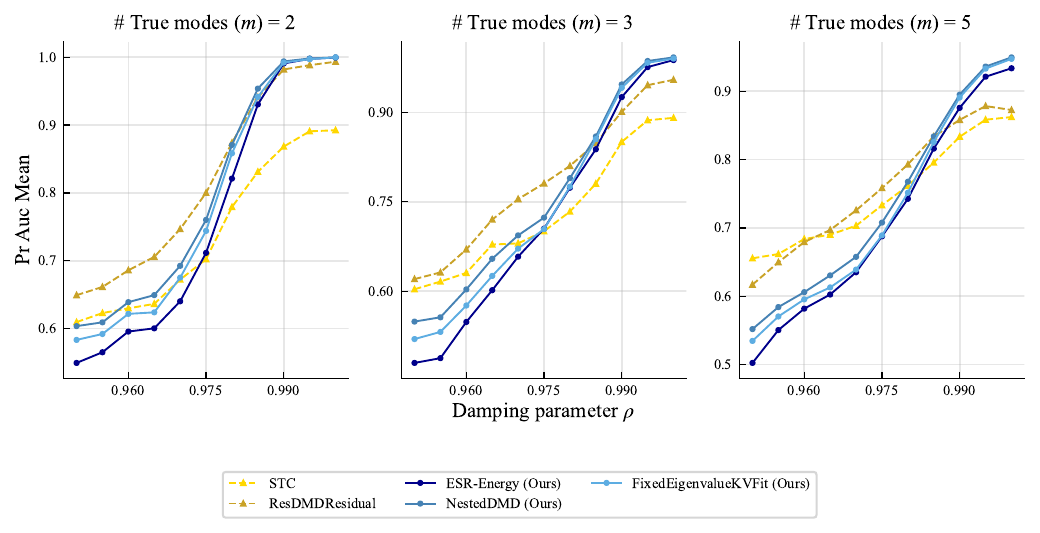}
   \caption{PR-AUC of mode classification vs.\ common eigenvalue magnitude
   $\rho$ at $D=2$ under Gaussian noise. Panels show $m=2$ (left), $m=3$
   (center), and $m=5$ (right). Working point matches
   Fig.~\ref{fig:pr_auc_snr};}
  \label{fig:pr_auc_rho}
\end{figure*}

Table~\ref{tab:pr_auc_unified} aggregates the mean PR-AUC of mode
selection across the four single-parameter sweeps and an additional
order-overestimation sweep $M$, per number of true modes $m$ and per spatial
dimension $D\in\{1,2,3\}$. The $D=2$ results correspond to the regime shown
in Figs.~\ref{fig:pr_auc_snr}--\ref{fig:pr_auc_kappa}; the $D=1$ and $D=3$
results summarize the same evaluation at smaller and larger spatial
dimension, with the corresponding figures available for viewing and
reproduction in the repository release~\cite{harris2026_dcdmd}. The cross-D results imply that the advantage of the \ac{KV}-aware scores over ResDMD is largest at $D=1$.
As $D$ increases and the delay-coordinates become less critical, the advantage of the \ac{KV}-aware scores becomes smaller.

\input{tables/pr_auc_unified.tex}

(Figs.~\ref{fig:pr_auc_snr}--\ref{fig:pr_auc_rho}), NestedDMD and
\ac{FEKVF} are at the top of the ordering, with \ac{ESR}-energy close
behind. The ResDMD residual is the strongest of the baselines: it tracks
\ac{ESR}-energy across most of the range and sits clearly above \ac{STC},
with one notable exception at low $\rho$, where both \ac{STC} and ResDMD
hold a small edge over our scores before the ordering reverses as $\rho \to
1$. The cross-$D$ trend in Table~\ref{tab:pr_auc_unified} is consistent
with the structural picture set out in
Sec.~\ref{subsec:resdmd_esr_companion}: the advantage of the
\ac{KV}-aware scores over ResDMD is largest at $D=1$, while at $D=3$
ResDMD is competitive across sweeps and even leads in the $\Delta\theta$
and $\rho$ scans.

\subsection{Concrete dynamical system: coupled masses on springs}
\label{subsec:masses_on_springs}

We complement the synthetic-mode evaluation of
Sec.~\ref{subsec:mode_classification_low_D} with a concrete
mass-spring-damper system from the family considered by
Bronstein~\emph{et al.}~\cite{Bronstein2022}: two masses, two springs, and
a wall-attached dashpot (Fig.~\ref{fig:two_mass_spring}). The system has
spatial dimension $D=2$ (one observation per mass) and produces $m=4$
underlying complex modes (two complex-conjugate oscillation pairs).

\input{figures/masses_on_springs.tex}

The equations of motion are
\begin{equation}
\bm M\,\ddot{\bm x} + \bm C\,\dot{\bm x} + \bm K\,\bm x = \bm 0,
\end{equation}
with
\[
\bm M = \begin{bmatrix} m_1 & 0 \\ 0 & m_2 \end{bmatrix},\quad
\bm C = \begin{bmatrix} c_1 & 0 \\ 0 & 0 \end{bmatrix},\quad
\bm K = \begin{bmatrix} k_1+k_2 & -k_2 \\ -k_2 & k_2 \end{bmatrix},
\]
producing weakly damped oscillatory modes whose eigenvalues lie just
inside the unit circle. The physical parameters $m_1, m_2, k_1, k_2, c_1$
determine the effective signal parameters $\rho$, $\Delta\theta$, and
$\kappa_b$ of the resulting modes; we therefore sweep over those effective
parameters (and \ac{SNR}) around a working point chosen to match the
simulated system, rather than over the physical parameters, which would
not directly inform the score-comparison question.

Figure~\ref{fig:masses_xparam_scan} reports the resulting multi-parameter
scans. On the \ac{SNR} and $\kappa_b$ axes, NestedDMD and \ac{FEKVF} outperform the baselines, with \ac{ESR}-energy close behind. ResDMD is weaker than ESR and stronger than \ac{STC}.
On $\Delta\theta$, all residual-based methods reach the same plateau quickly, with
On $\rho$, \ac{STC} and ResDMD perform better than our scores at small
$\rho$, with the ordering reversing as $\rho\to 1$. This is structurally
natural: when damping is strong, the true modes decay rapidly and the
effective support of the Vandermonde sequence $\bm{v}_L(\widehat{\lambda}_j)$
is short and only the early lags carry appreciable energy, and the \ac{KV}
template becomes a less discriminative criterion across the full lag window.

\begin{figure*}[t]
  \centering
  \includegraphics[width=\linewidth]{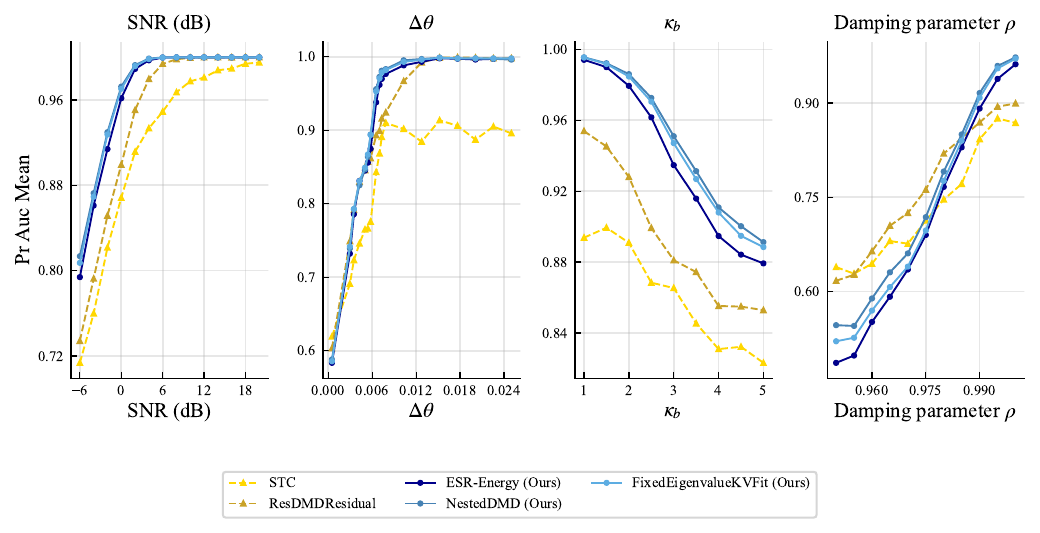}
   \caption{PR-AUC of mode classification for the coupled-mass-on-springs
   system with two masses ($D=2$, $m=4$), under Gaussian noise. Panels
   report sweeps over \ac{SNR}, phase separation $\Delta\theta$, amplitude
   heterogeneity $\kappa_b$, and common eigenvalue magnitude $\rho$. The
   setup follows Bronstein~\emph{et al.}~\cite{Bronstein2022}.}
  \label{fig:masses_xparam_scan}
\end{figure*}

\subsection{Spurious eigenvalue behavior under delay-coordinates}
\label{subsec:spurious_eigenvalue_behavior}

To test the effect of the embedding length $L$ on spurious eigenvalues, we fix $m = 3$ and $M = 15$ and run 500 trials.
For each $L$, we pool the magnitudes $|\widehat{\lambda}_j|$ of all spurious eigenvalues across all trials and compute the empirical CDF.

Figure~\ref{fig:spur_eigs_L} shows that the distribution of spurious eigenvalue magnitudes shifts upward with $L$, increasingly concentrating near the unit circle. At the same time, the pooled distributions retain a non-negligible lower tail for all tested $L$, indicating substantial variability even when typical magnitudes are large.

\begin{figure}[t]
  \centering
  \includegraphics[width=\columnwidth]{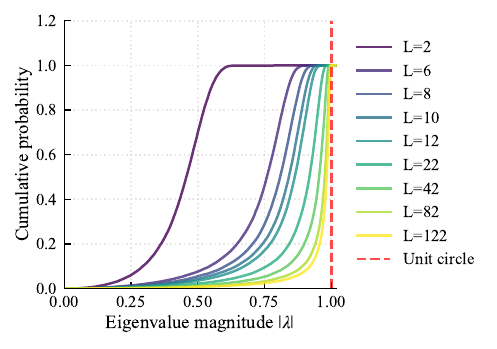}
    \caption{Spurious eigenvalue magnitude statistics versus embedding length $L$.
    For each $L$, we pool the magnitudes $|\widehat{\lambda}|$ of all spurious eigenvalues from all Monte--Carlo trials and plot the resulting empirical CDF (one curve per $L$). Overall, spurious eigenvalue magnitudes shift upward with increasing $L$, but the pooled distributions retain a non-negligible lower tail.}
      \label{fig:spur_eigs_L}
\end{figure}

This is the empirical observation highlighted in Sec.~\ref{subsec:mode_norm_analysis}: as $L$ increases, the bulk of spurious
eigenvalue magnitudes shifts toward the unit circle, while a substantial
lower tail persists. Consequently, even at large $L$, eigenvalue magnitude
alone cannot reliably separate true and spurious components.

%% file: tables/pr_auc_unified.tex
\begin{table*}[t]
    \centering
    \scriptsize
    \setlength{\tabcolsep}{2pt}
    \begin{tabular}{llcccccccccccccccc}
    \hline
     &  & \multicolumn{5}{c}{$m = 2$} & \multicolumn{5}{c}{$m = 3$} & \multicolumn{5}{c}{$m = 5$} &  \\
    $D$ & Method & SNR & $\Delta\theta$ & $r$ & $\kappa_b$ & $M$ & SNR & $\Delta\theta$ & $r$ & $\kappa_b$ & $M$ & SNR & $\Delta\theta$ & $r$ & $\kappa_b$ & $M$ & \# Wins \\
    \hline
    1 & STC & 0.931 & 0.864 & 0.753 & 0.859 & 0.832 & 0.912 & 0.849 & \textbf{0.733} & 0.812 & 0.771 & 0.881 & 0.849 & \textbf{0.757} & 0.809 & 0.750 & 2 \\
     & ResDMDResidual & 0.912 & 0.897 & 0.669 & 0.786 & 0.779 & 0.819 & 0.847 & 0.628 & 0.665 & 0.527 & 0.735 & 0.816 & 0.627 & 0.642 & 0.466 & 0 \\
     & ESR-Energy & 0.964 & 0.935 & 0.734 & 0.910 & 0.969 & 0.931 & 0.902 & 0.681 & 0.826 & 0.812 & 0.875 & 0.885 & 0.681 & 0.813 & 0.777 & 0 \\
     & NestedDMD & \textbf{0.973} & \textbf{0.940} & \textbf{0.775} & \textbf{0.935} & \textbf{0.980} & \textbf{0.941} & \textbf{0.912} & 0.721 & \textbf{0.851} & \textbf{0.861} & \textbf{0.891} & \textbf{0.895} & 0.719 & \textbf{0.835} & \textbf{0.820} & 13 \\
     & FixedEigenvalueKVFit & 0.967 & 0.937 & 0.736 & 0.917 & 0.969 & 0.937 & 0.910 & 0.689 & 0.841 & 0.837 & 0.887 & 0.894 & 0.693 & 0.830 & 0.806 & 0 \\
    \hline
    2 & STC & 0.930 & 0.864 & 0.738 & 0.871 & 0.839 & 0.927 & 0.855 & 0.730 & 0.863 & 0.828 & 0.923 & 0.862 & 0.748 & 0.854 & 0.793 & 0 \\
     & ResDMDResidual & 0.982 & \textbf{0.970} & \textbf{0.821} & 0.959 & 0.988 & 0.963 & 0.943 & \textbf{0.785} & 0.918 & 0.938 & 0.940 & 0.924 & \textbf{0.762} & 0.873 & 0.832 & 4 \\
     & ESR-Energy & 0.990 & 0.969 & 0.763 & 0.978 & \textbf{0.999} & 0.977 & 0.946 & 0.725 & 0.949 & 0.983 & 0.963 & 0.933 & 0.713 & 0.926 & 0.926 & 1 \\
     & NestedDMD & \textbf{0.992} & \textbf{0.970} & 0.797 & \textbf{0.984} & 0.998 & \textbf{0.982} & \textbf{0.948} & 0.758 & \textbf{0.961} & \textbf{0.989} & \textbf{0.969} & \textbf{0.937} & 0.737 & \textbf{0.938} & \textbf{0.947} & 11 \\
     & FixedEigenvalueKVFit & 0.991 & \textbf{0.970} & 0.784 & 0.981 & 0.998 & 0.980 & \textbf{0.948} & 0.742 & 0.957 & 0.987 & 0.968 & \textbf{0.937} & 0.725 & 0.936 & 0.942 & 3 \\
    \hline
    3 & STC & 0.938 & 0.870 & 0.747 & 0.879 & 0.853 & 0.933 & 0.857 & 0.737 & 0.877 & 0.835 & 0.930 & 0.869 & 0.749 & 0.872 & 0.813 & 0 \\
     & ResDMDResidual & 0.995 & \textbf{0.980} & \textbf{0.871} & 0.991 & \textbf{1.000} & 0.986 & \textbf{0.960} & \textbf{0.849} & 0.970 & 0.993 & 0.973 & \textbf{0.951} & \textbf{0.818} & 0.948 & 0.961 & 7 \\
     & ESR-Energy & 0.996 & 0.977 & 0.773 & 0.993 & \textbf{1.000} & 0.989 & 0.954 & 0.742 & 0.977 & 0.998 & 0.979 & 0.946 & 0.726 & 0.961 & 0.979 & 1 \\
     & NestedDMD & \textbf{0.997} & 0.977 & 0.813 & \textbf{0.996} & \textbf{1.000} & \textbf{0.992} & 0.955 & 0.779 & \textbf{0.986} & \textbf{0.999} & \textbf{0.983} & 0.946 & 0.749 & \textbf{0.971} & \textbf{0.990} & 9 \\
     & FixedEigenvalueKVFit & \textbf{0.997} & 0.978 & 0.807 & 0.995 & \textbf{1.000} & \textbf{0.992} & 0.955 & 0.773 & 0.984 & \textbf{0.999} & \textbf{0.983} & 0.947 & 0.748 & 0.969 & 0.989 & 5 \\
    \hline
    \end{tabular}
    \caption{Normalized PR-AUC for Gaussian-noise DC experiments across one-dimensional parameter sweeps at $\mathrm{SNR}=0\,\mathrm{dB}$, $\Delta\theta=0.007$, $r=1.0$, $\kappa_b=2.5$, $M=15$ (except $M$ sweep, in which $M$ is varied over $[10,40]$).}
    \label{tab:pr_auc_unified}
    \end{table*}

%% file: figures/masses_on_springs.tex
\begin{figure}[t]
\centering

\begin{tikzpicture}[
    mass/.style={
        draw,
        thick,
        minimum width=1.4cm,
        minimum height=1.2cm,
        fill=white
    },
    spring/.style={
        thick,
        decorate,
        decoration={
            coil,
            aspect=0.42,
            segment length=4pt,
            amplitude=4pt
        }
    },
    disp/.style={
        -{Latex[length=2mm]},
        thick
    }
]

\def\groundy{0}
\def\massCenterY{1.0}
\def\springAtY{1.15}
\def\damperAtY{0.55}
\def\damperH{0.13}
\def\wheelY{0.20}
\def\wheelR{0.20}

\draw[thick] (0,0) -- (0,2.0);
\fill[pattern=north east lines] (-0.35,0) rectangle (0,2.0);

\draw[thick] (0,\groundy) -- (5.8,\groundy);
\fill[pattern=north east lines] (0,-0.25) rectangle (5.8,\groundy);

\node[mass] (m1) at (2.4,\massCenterY) {$m_1$};
\node[mass] (m2) at (4.7,\massCenterY) {$m_2$};

\foreach \x in {2.1, 2.7} {
    \fill (\x,\wheelY) circle (\wheelR);
}
\foreach \x in {4.4, 5.0} {
    \fill (\x,\wheelY) circle (\wheelR);
}

\draw[spring] (0,\springAtY) -- (1.7,\springAtY);
\draw[spring] (3.1,\springAtY) -- (4.0,\springAtY);

\node[above] at (0.85,1.30) {$k_1$};
\node[above] at (3.55,1.30) {$k_2$};

\def\damperCylRightX{1.00}
\def\damperPistonX{0.75}
\draw[thick] (0,\damperAtY+\damperH) -- (\damperCylRightX,\damperAtY+\damperH);
\draw[thick] (0,\damperAtY-\damperH) -- (\damperCylRightX,\damperAtY-\damperH);
\draw[thick] (0,\damperAtY-\damperH) -- (0,\damperAtY+\damperH);
\draw[thick] (\damperPistonX,\damperAtY+\damperH-0.01) -- (\damperPistonX,\damperAtY-\damperH+0.01);
\draw[thick] (\damperPistonX,\damperAtY) -- (1.7,\damperAtY);
\node[below] at (0.55,\damperAtY-\damperH) {$c_1$};

\draw[disp] (2.0,1.85) -- (2.8,1.85);
\node[above] at (2.4,1.85) {$x_1(t)$};

\draw[disp] (4.3,1.85) -- (5.1,1.85);
\node[above] at (4.7,1.85) {$x_2(t)$};

\draw[thick] (2.0,1.70) -- (2.0,2.00);
\draw[thick] (4.3,1.70) -- (4.3,2.00);

\end{tikzpicture}

\caption{Two-mass spring system, comprised of two masses $m_1$ and $m_2$, connected via two springs ($k_1, k_2$) and damped by a dashpot ($c_1$) between the wall and $m_1$. Each real oscillation corresponds to two complex oscillations, and hence the total order is $m=4$.}
\label{fig:two_mass_spring}

\end{figure}

%% file: sections/discussion.tex
\section{Conclusions}
\label{sec:discussion}

A central finding of this work concerns the internal organization that delay coordinates impose on the computed modes, and the role it plays in mode selection. Through an operator-theoretic analysis of the delay embedding via the block-companion least-squares propagator, we show that every computed mode (true or spurious) inherits a Kronecker--Vandermonde (KV) structure across delays. True modes are distinguished not by the presence of this structure, but by the degree to which they conform to it. The deviation from the ideal KV form is governed precisely by the geometric residual associated with the signal-subspace viewpoint, the estimated-subspace residual (ESR). Mode-internal KV structure is therefore the primary discriminative signal in scenarios where delay coordinates are essential, and the ESR is the geometric quantity that controls it.

An additional conceptual message of our analysis concerns the relationship between \ac{DMD} and the unconstrained least-squares (LS) problem used to approximate temporal evolution. By exploiting the tight connection between truncation in \ac{DMD} and rank-constrained LS formulations, we show that truncation manifests geometrically as a deviation of the computed modes from the ideal \ac{KV} structure induced by delay-coordinates. Moreover, by linking this structural deviation to signal-subspace geometry through the \ac{ESR}, we demonstrate that it is not uniform across modes: signal-consistent (true) modes adhere more strongly to the \ac{KV} structure, while spurious modes exhibit systematically larger departures. This perspective clarifies how finite-rank approximation reshapes the mode geometry and provides a unified interpretation of residual- and structure-based selection criteria.
In addition, it reveals that truncation in \ac{DMD} is not merely a numerical device, but a structural operation that reshapes the companion-induced organization of the modes in a mode-dependent manner.

A practical consideration is sensitivity to the hyper-parameters $L$ and $M$.
The embedding length $L$ should be selected for modeling and spectral
estimation, not for detection. Empirically, our methods are effective from
modest $L$ values and do not deteriorate as $L$ grows. 
The truncation rank $M$ is best interpreted through the overestimation gap $M-m$: while $M$ must exceed
$m$ to allow post hoc mode selection, overly large $M$ shrinks residual energy
for all modes and can reduce the discriminative power of the \ac{ESR}. 
This sensitivity suggests that richer geometric descriptors, such as the distribution of mode energy across singular directions, may provide more stable information than an aggregated residual summary alone.
In contrast, the \ac{KV}-deviation criteria largely mitigate the dependence
on $M$, remaining reliable even under aggressive overestimation.

Several future directions appear promising. First, it would be valuable to characterize
how true and spurious modes distribute energy across singular directions and to
develop selection rules that exploit distributional differences directly.
Second, we observe that spurious spectral statistics vary systematically with
the embedding length $L$; explaining the mechanism behind this dependence would
deepen our understanding of how delay-coordinates reshape the spurious
spectrum. Finally, although the \ac{ESR} and \ac{KV} deviations are conceptually linked through
the companion eigenrelation (Theorem~\ref{thm:residual_companion_identity}),
developing a unified theory that quantitatively links residual geometry, \ac{KV} conformity, and spectral statistics under delay embedding remains an important open direction.

%% file: sections/data_availability.tex
\section*{Data Availability Statement}
The data that support the findings of this study are available within the article. The data were generated via numerical simulations, and the Python code used to generate the data and reproduce the figures is available on GitHub at \url{https://github.com/YoavHarris/geometric-dc-dmd-order-detection}. A citable archival release is available via Zenodo \cite{harris2026_dcdmd}.

%% file: appendices/wide_regime.tex
\section[Relations between D, L and N]{Relations between \texorpdfstring{$D$}{D}, \texorpdfstring{$L$}{L} and \texorpdfstring{$N$}{N}}
\label{app:wide_matrix}

This appendix revisits several identities from the main text in the
\emph{wide} delay-embedded setting, where the snapshot matrix
$\bm{X}_0 \in \mathbb{C}^{DL\times (N-L)}$ satisfies
\begin{equation}
  DL < N - L .
\end{equation}
In this regime, $\bm{X}_0$ has more columns than rows.  Throughout we
assume that $\bm{X}_0$ has full row rank,
\begin{equation}
  \operatorname{rank}(\bm{X}_0)=DL,
\end{equation}
which is a generic outcome under sufficiently rich excitation or
additive noise.

Full row rank implies that $\bm{X}_0\bm{X}_0^{\mathrm{H}}$ is invertible and that
the Moore--Penrose pseudoinverse admits the form
\begin{equation}
  \bm{X}_0^\dagger \;=\; \bm{X}_0^{\mathrm{H}}\,(\bm{X}_0\bm{X}_0^{\mathrm{H}})^{-1}.
\end{equation}
Let $\mathcal U=\operatorname{col}(\bm X_0)$ and denote by
$\bm P_{\mathcal U}$ the orthogonal projector onto $\mathcal U$.  Since
$\operatorname{rank}(\bm X_0)=DL$, we have $\mathcal U=\mathbb{C}^{DL}$
and therefore
\begin{equation}
  \bm P_{\mathcal U}=\bm I_{DL}.
\end{equation}
Consequently, the least-squares propagator
\begin{equation}
  \bm{A}_{\mathrm{MP}} \;=\; \bm{X}_1\bm{X}_0^\dagger
\end{equation}
is uniquely defined.

This identity has immediate implications for residual-based diagnostics.
Because $\mathcal U=\mathbb{C}^{DL}$, there is no component outside
$\mathcal U$; in particular, for an exact DMD mode
$\widehat{\bm{\phi}}^{\,e}$ the decomposition
\begin{equation}
  \widehat{\bm{\phi}}^{\,e}
  \;=\;
  \bm{P}_{\mathcal U_M}\widehat{\bm{\phi}}^{\,e}
  +(\bm{P}_{\mathcal U}-\bm{P}_{\mathcal U_M})\widehat{\bm{\phi}}^{\,e}
  +(\bm{I}-\bm{P}_{\mathcal U})\widehat{\bm{\phi}}^{\,e}
\end{equation}
reduces to
\begin{equation}
  \widehat{\bm{\phi}}^{\,e}
  \;=\;
  \bm{P}_{\mathcal U_M}\widehat{\bm{\phi}}^{\,e}
  +(\bm{I}-\bm{P}_{\mathcal U_M})\widehat{\bm{\phi}}^{\,e},
\end{equation}
because $(\bm{I}-\bm{P}_{\mathcal U})\widehat{\bm{\phi}}^{\,e}=\bm 0$.
Thus, residual-type scores relative to $\mathcal U_M$ quantify only the
truncation effect $(\bm{I}-\bm{P}_{\mathcal U_M})\widehat{\bm{\phi}}^{\,e}$,
which shrinks monotonically as $M$ increases toward $\min(DL,N-L)=DL$.
In other words, the wide full-row-rank regime is an extreme case in which
$\bm P_{\mathcal U}=\bm I_{DL}$ and residual-type scores
relative to $\mathcal U_M$ reflect only the truncation term
$(\bm I-\bm P_{\mathcal U_M})\widehat{\bm\phi}^{\,e}$, which shrinks to zero as
$M\uparrow DL$.  We therefore focus on the more common tall regime
$N-L<DL$, which is typical in high-dimensional measurements (and even
more so under delay coordinates) and is also more challenging from an
estimation standpoint since fewer delay vectors are available relative to the
number of unknown parameters.

The same simplification carries over to the block-companion structure.
In the delay-embedded construction, the shift relations in the first
$L-1$ block rows hold exactly by definition, so the
block-companion matrix $\bm C_L$ is feasible for the constrained
least-squares problem that enforces these shift equations.  In the wide
full-row-rank regime, the identity $\bm A_{\mathrm{MP}}=\bm C_L\bm P_{\mathcal U}$
simplifies (since $\bm P_{\mathcal U}=\bm I_{DL}$) to
\begin{equation}
  \bm A_{\mathrm{MP}}=\bm C_L .
\end{equation}
Hence the Moore--Penrose solution already acts as an exact shift on all
$DL$ coordinates, and no additional projection through $\bm P_{\mathcal U}$
takes place.

%% file: appendices/choose_m_singular_vectors.tex
\section[Choosing the signal-aligned subspace]{Choosing the signal-aligned subspace}
\label{app:choose_m_vecs_out_of_M}

The best $m$-dimensional proxy to $\mathcal S$ \emph{within} $\mathcal U_M$
is not necessarily given by the span of the leading $m$ singular vectors.
More generally, $\mathcal U_m$ may be defined as the span of the $m$
singular vectors (among the leading $M$) that minimize the projector
distance (cf.\ \eqref{eq:eta_bound}) to $\mathcal S$:
\begin{equation}
\mathcal U_m \coloneqq \operatorname{col}(\bm U_{\mathcal I_m}),
\label{eq:Um_def}
\end{equation}
where $\bm U_{\mathcal I}$ denotes the matrix consisting of the singular
vectors $\{\bm u_i\}_{i\in\mathcal I}$ for an index set
$\mathcal I\subset\{1,\dots,M\}$. The optimal index set is defined by
\begin{equation}
\mathcal I_m \in
\arg\min_{\mathcal I\subset\{1,\dots,M\},\,|\mathcal I|=m}
\bigl\|\bm P_{\mathcal S}
-
\bm P_{\operatorname{col}(\bm U_{\mathcal I})}
\bigr\|_2,
\label{eq:Im_def}
\end{equation}
where $\bm P_{\operatorname{col}(\bm U_{\mathcal I})}$ denotes the
orthogonal projector onto $\operatorname{col}(\bm U_{\mathcal I})$.

This construction is conceptual, since $\mathcal S$, and hence the optimal
index set $\mathcal I_m$, are unknown in practice. It serves to formalize
the ideal signal-consistent subspace against which computed modes are
compared.

%% file: appendices/bound_on_subspace_deviation.tex
\section{Bound on signal--subspace deviation under noise}
\label{app:signal_subspace_identifiability}

This appendix specializes classical subspace perturbation theory to the
delay--embedded trajectory matrix and quantifies the deviation between the
true signal subspace
\[
\mathcal{S}=\operatorname{col}(\bm{S}_0)
\]
and its empirical estimate $\operatorname{col}(\bm{U}_m)$, where $\bm{U}_m$
contains the leading $m$ left singular vectors of
\[
\bm{X}_0=\bm{S}_0+\bm{E}.
\]

\subsection*{Model and notation}

The clean signal follows the exponential model
\[
\bm{s}_k=\sum_{j=1}^{m} b_j\,\bm{\phi}_j\,\lambda_j^{k},
\qquad
\lambda_j=\rho_j\,\mathrm{e}^{\mathrm{i}\theta_j},
\]
with spatial vectors $\bm{\phi}_j\in\mathbb{C}^D$. We assume a bounded-radius
regime $0<\rho_{\min}\le \rho_j\le \rho_{\max}\le 1$.

Let $N$ denote the number of samples $\{\bm{s}_k\}_{k=0}^{N-1}$.
After delay embedding with $L$ delays, the number of usable delay vectors is
\[
N_{\mathrm{eff}}:=N-L+1,
\]
and the standard DC--DMD snapshot matrices satisfy
\[
\bm{X}_0,\bm{X}_1\in\mathbb{C}^{DL\times (N-L)}.
\]
In particular, the noiseless part of $\bm{X}_0$ admits the factorization
\[
\bm{S}_0
=
\widetilde{\bm{\Phi}}_L\,\bm{B}\,\bm{V}_{N-L}^{\top},
\]
where the lifted spatial factor is
\[
\widetilde{\bm{\Phi}}_L
:=
\bm{V}_L\odot\bm{\Phi}
=
\bigl[
\bm{v}_L(\lambda_1)\otimes\bm{\phi}_1\ \cdots\
\bm{v}_L(\lambda_m)\otimes\bm{\phi}_m
\bigr]
\in\mathbb{C}^{DL\times m},
\]
with base spatial factor $\bm{\Phi}=[\bm{\phi}_1\cdots\bm{\phi}_m]$. The amplitudes are given by $\bm{B}=\operatorname{diag}(b_1,\ldots,b_m)$, and the Vandermonde matrix is defined generally as

\[
\begin{aligned}
\bm{V}_{n}
&=
\bigl[\bm{v}_{n}(\lambda_1)\ \cdots\ \bm{v}_{n}(\lambda_m)\bigr]
\in\mathbb{C}^{n\times m}, \\
\bm{v}_{n}(\lambda)
&=
[\,1,\lambda,\ldots,\lambda^{n-1}\,]^{\top}.
\end{aligned}
\]

Let $\theta_{\max}$ denote the largest principal angle between
$\operatorname{col}(\bm{U}_m)$ and $\mathcal{S}$.

\subsection*{Assumptions}

We assume:
\begin{itemize}
\item[(A1)] \emph{Phase separation}:
\[
\Delta_{\theta}
:=
\min_{j\neq j'}
\min\bigl(|\theta_j-\theta_{j'}|,\,2\pi-|\theta_j-\theta_{j'}|\bigr)>0.
\]
\item[(A2)] \emph{Nondegeneracy}:
$\sigma_m(\bm{\Phi})>0$ and $\min_j |b_j|>0$.
\item[(A3)] $L$ and $N-L$ exceed a constant multiple of $1/\Delta_{\theta}$.
\end{itemize}

\subsection*{Signal--subspace deviation}

\begin{lemma}[Signal--subspace deviation]
\label{lem:signal_subspace_dev}
If $\|\bm{E}\|_2$ is sufficiently small, then
\[
\|
\bm{P}_{\operatorname{col}(\bm{U}_m)}
-
\bm{P}_{\mathcal{S}}
\|_2
=
\sin\theta_{\max}
\le
\eta,
\]
where
\[
\eta
=
\frac{\|\bm{E}\|_2}{
\sigma_m(\bm{\Phi})\,\min_j |b_j|\,
\alpha_L\,\alpha_{N-L}
-
\|\bm{E}\|_2}.
\]
Here $\alpha_n$ denotes any lower bound on $\sigma_m(\bm{V}_n)$ valid under
assumptions (A1)--(A3).
\end{lemma}

\begin{proof}[Sketch]
Wedin's $\sin\Theta$ theorem gives
\[
\sin\theta_{\max}
\le
\frac{\|\bm{E}\|_2}{
\sigma_m(\bm{S}_0)-\|\bm{E}\|_2}.
\]
Using the factorization of $\bm{S}_0$,
\[
\sigma_m(\bm{S}_0)
\ge
\sigma_m(\widetilde{\bm{\Phi}}_L)\,
\min_j|b_j|\,
\sigma_m(\bm{V}_{N-L}).
\]
Since $\widetilde{\bm{\Phi}}_L=\bm{V}_L\odot\bm{\Phi}$,
\[
\sigma_m(\widetilde{\bm{\Phi}}_L)
\ge
\sigma_m(\bm{\Phi})\,\sigma_m(\bm{V}_L).
\]
Combining the inequalities and using $\sigma_m(\bm{V}_n)\ge\alpha_n$ yields the
result.
\end{proof}

\subsection*{Vandermonde conditioning}

Under~(A1), classical results on Vandermonde conditioning imply that
\[
\sigma_m(\bm{V}_n)\ge \alpha_n,
\]
for some $\alpha_n>0$ depending on the minimal phase gap $\Delta_\theta$, the
order $m$, and the bounded-radius regime $[\rho_{\min},\rho_{\max}]$.
Smaller phase gaps worsen the conditioning and reduce $\alpha_n$.
When $\rho_{\max}=1$ (unit-modulus nodes), one may take a bound of the form
\[
\alpha_n \ge c_{\mathrm{sep}}(\Delta_\theta,m)\sqrt{n},
\]
whereas for $\rho_{\max}<1$ the dependence on $n$ need not exhibit the same
$\sqrt{n}$ growth and may level off, with constants controlled additionally by
$[\rho_{\min},\rho_{\max}]$.
See, e.g.,~\cite{Candes2014,Kunis2016}.

\subsection*{Discussion}

The deviation $\sin\theta_{\max}$ scales linearly with $\|\bm{E}\|_2$ and is
controlled by the conditioning of the lifted signal factors. The dependence on
$L$ and $N-L$ enters through the Vandermonde conditioning constants $\alpha_L$
and $\alpha_{N-L}$; when $\rho_{\max}=1$ this recovers the familiar
$\sqrt{L\,(N-L)}$ scaling, while for $\rho_{\max}<1$ the improvement with $n$ may
saturate. The result provides a quantitative justification for
Assumption~\ref{ass:subspace_distance_bound} in the main text.

%% file: appendices/concise_spatiotemporal_coupling.tex
\section{Delay-coordinates yield Kronecker--Vandermonde spatiotemporal modes}
\label{app:data_side_kv}

This appendix gives a short data-side derivation of the structural result
identified in Ref.~\onlinecite{Bronstein2022}: a signal that is a sum of spatially
modulated exponentials in time becomes, after delay-coordinates, a sum of
spatiotemporally coupled \ac{KV} templates, each still
modulating a temporal exponential across the embedded columns.
The argument requires no operator assumptions and applies only to the true
components, providing a concise complement to the operator-side derivation in
Sec.~\ref{sec:kv_structure}, which extends the \ac{KV} characterization to
\emph{all} modes (true and spurious).

\subsection*{Signal model}

We adopt the exponential signal model from~\eqref{eq:explicit_signal_structure}.
For an embedding length $L$, collect $L$ successive samples into
\[
\widetilde{\bm{s}}_k
=
\begin{bmatrix}
\bm{s}_k\\
\bm{s}_{k+1}\\
\vdots\\
\bm{s}_{k+L-1}
\end{bmatrix}
\in\mathbb{C}^{DL}.
\]

\subsection*{Embedding of a single exponential component}

For the $j$th component $b_j\,\bm{\phi}_j\,\lambda_j^{k}$,
\[
\widetilde{\bm{s}}_{k}^{(j)}
=
b_j
\begin{bmatrix}
\bm{\phi}_j\,\lambda_j^{k}\\
\bm{\phi}_j\,\lambda_j^{k+1}\\
\vdots\\
\bm{\phi}_j\,\lambda_j^{k+L-1}
\end{bmatrix}
=
b_j\,
\bigl(\bm{v}_L(\lambda_j)\otimes\bm{\phi}_j\bigr)\,\lambda_j^{k},
\]
where $\bm{v}_L(\lambda_j)=[1,\lambda_j,\ldots,\lambda_j^{L-1}]^{\!\top}$.
Delay-coordinates therefore fuse the spatial vector $\bm{\phi}_j$ and temporal
exponential $\lambda_j^{k}$ into the \ac{KV} template
$\bm{v}_L(\lambda_j)\otimes\bm{\phi}_j$.

\subsection*{\ac{KV} decomposition of the embedded trajectory matrix}

Let $N_{\mathrm{eff}}=-L+1$ be the number of embedded columns.  The trajectory
matrix becomes
\[
\widetilde{\bm{S}}
=
\sum_{j=1}^{m}
b_j\,
\bigl(\bm{v}_L(\lambda_j)\otimes\bm{\phi}_j\bigr)
\bigl[1,\lambda_j,\ldots,\lambda_j^{N_{\mathrm{eff}}-1}\bigr].
\]
Writing $\bm{v}_{N_{\mathrm{eff}}}(\lambda_j)$ for the temporal Vandermonde
vector, this is
\[
\widetilde{\bm{S}}
=
\sum_{j=1}^{m}
b_j\,
\bigl(\bm{v}_L(\lambda_j)\otimes\bm{\phi}_j\bigr)
\bm{v}_{N_{\mathrm{eff}}}(\lambda_j)^{\!\top}.
\]

\subsection*{Conclusion}

Under delay-coordinates, every true component becomes a \ac{KV} term whose
spatial and temporal factors are inseparably coupled.  These are the ideal \ac{KV} templates assumed in Bronstein \emph{et al.}\cite{Bronstein2022}.
The operator-side analysis of Sec.~\ref{sec:kv_structure} generalizes this
perspective by showing that \emph{all} projected-\ac{DMD} modes,
including spurious ones, satisfy an approximate \ac{KV} relation whose deviation is
measured by their estimated-subspace residual.

%% file: appendices/block_companion.tex
\section{Block-Companion Structure and Least-Squares Propagators}
\label{app:companion_proofs}

We consider the least-squares problem
\[
\min_{\bm{A}\in\mathbb{C}^{DL\times DL}}
\|\bm{A}\bm{X}_0 - \bm{X}_1\|_F^2,
\]
with delay-embedded snapshot matrices
$\bm{X}_0,\bm{X}_1\in\mathbb{C}^{DL\times N_{\mathrm{eff}}}$.
Among all minimizers we single out the Moore–Penrose (minimum-norm) minimizer
\[
\bm{A}_{\mathrm{MP}} := \bm{X}_1\bm{X}_0^\dagger,
\]
which will be related to the reduced propagator via an orthogonal projection.

\subsection{Existence of a block-companion minimizer}
\label{app:existence_block_companion}

Partition $\bm{A}$ into $D\times D$ block rows
\[
\bm{A}=
\begin{bmatrix}
\bm{A}^{(0)}\\[-1mm]\vdots\\[-1mm]\bm{A}^{(L-1)}
\end{bmatrix},
\qquad
\bm{A}^{(\ell)}\in\mathbb{C}^{D\times DL}.
\]
The objective splits as
\[
\|\bm{A}\bm{X}_0-\bm{X}_1\|_F^2
=\sum_{\ell=0}^{L-1}\|\bm{A}^{(\ell)}\bm{X}_0-\bm{X}_1^{(\ell)}\|_F^2.
\]

In \ac{DCDMD},
\[
\bm{X}_1^{(\ell)}=\bm{S}_L^{(\ell)}\bm{X}_0,\qquad \ell=0,\ldots,L-2,
\]
for the block-shift matrix $\bm{S}_L$.
Thus setting $\bm{A}^{(\ell)}=\bm{S}_L^{(\ell)}$ for $\ell<L-1$ is already optimal.
Minimizing only the last block row yields some $\bm{B}$, giving the
block–companion minimizer
\[
\bm{C}_L=\bm{S}_L+(\bm{e}_L\!\otimes\!\bm{I}_D)\bm{B}.
\]

Minimizing only the last block row amounts to
\[
\min_{\bm B\in\mathbb C^{D\times DL}}\ \|\bm B\bm X_0-\bm X_1^{(L-1)}\|_F^2.
\]
Throughout, we fix
\[
\bm B:=\bm X_1^{(L-1)}\bm X_0^\dagger,
\]
the Moore--Penrose least-squares predictor.

\subsection[Eigenvectors of the block-companion matrix are KV]{Eigenvectors of the block-companion matrix are \ac{KV}}
\label{app:kv_eigenvectors}

Let
\[
\bm{C}_L=
\begin{bmatrix}
\bm{0} & \bm{I}_D & & \\
 & \ddots & \ddots & \\
 & & \bm{0} & \bm{I}_D\\
\bm{B}_1 & \bm{B}_2 & \cdots & \bm{B}_L
\end{bmatrix}.
\]
Write an eigenvector as
\[
\bm{v}=[\,\bm{v}^{(0)};\ldots;\bm{v}^{(L-1)}\,].
\]
For $\ell=0,\ldots,L-2$,
\[
\bm{v}^{(\ell+1)}=\lambda\,\bm{v}^{(\ell)},
\]
so $\bm{v}^{(\ell)}=\lambda^\ell\bm{v}^{(0)}$.
The final block row imposes
\[
(\bm{B}_1+\lambda\bm{B}_2+\cdots+\lambda^{L-1}\bm{B}_L)\bm{v}^{(0)}
=\lambda^L\bm{v}^{(0)}.
\]
Hence every eigenvector has \ac{KV} form
\[
\bm{v}=\bm{v}_L(\lambda)\otimes\bm{v}^{(0)},
\qquad
\bm{v}_L(\lambda)=[\,1,\lambda,\ldots,\lambda^{L-1}\,]^{\!\top}.
\]

\subsection[The Moore-Penrose minimizer satisfies X1 X0+ = CL PU]{The Moore–Penrose minimizer satisfies \texorpdfstring{$\bm{X}_1\bm{X}_0^\dagger = \bm{C}_L \bm{P}_{\mathcal U}$}{X1 X0+ = CL PU}}
\label{app:moore_penrose_projection}

Let $\mathcal{U}=\operatorname{col}(\bm{X}_0)$ and $\bm{P}_{\mathcal U}$ its projector.

\begin{lemma}[Moore--Penrose projection identity]
\label{lem:moore_penrose_projection_core}
Let $\mathcal U=\operatorname{col}(\bm X_0)$ and let
\[
\bm P_{\mathcal U}:=\bm X_0\bm X_0^\dagger
\]
be the orthogonal projector onto $\mathcal U$.
Consider the least-squares problem
\[
\min_{\bm A\in\mathbb C^{DL\times DL}}\;\|\bm A\bm X_0-\bm X_1\|_F^2,
\]
and define the Moore--Penrose (minimum-norm) minimizer
\[
\bm A_{\mathrm{MP}}:=\bm X_1\bm X_0^\dagger.
\]
Then, for any minimizer $\bm A_\star$,
\[
\bm A_{\mathrm{MP}}=\bm A_\star\,\bm P_{\mathcal U}.
\]
\end{lemma}

\begin{proof}
Any minimizer $\bm A_\star$ satisfies the normal equations
\[
(\bm A_\star\bm X_0-\bm X_1)\bm X_0^{\!H}=\bm 0,
\]
equivalently,
\[
\bm A_\star\bm X_0\bm X_0^{\!H}=\bm X_1\bm X_0^{\!H}.
\]
A standard characterization of the solution set of this Frobenius least-squares
problem is
\[
\bm A_\star=\bm X_1\bm X_0^\dagger+\bm Z(\bm I-\bm P_{\mathcal U}),
\qquad \bm Z\in\mathbb C^{DL\times DL}.
\]
Right-multiplying by $\bm P_{\mathcal U}$ and using
$\bm P_{\mathcal U}(\bm I-\bm P_{\mathcal U})=\bm 0$ yields
\[
\bm A_\star\bm P_{\mathcal U}
=(\bm X_1\bm X_0^\dagger+\bm Z(\bm I-\bm P_{\mathcal U}))\bm P_{\mathcal U}
=\bm X_1\bm X_0^\dagger
=\bm A_{\mathrm{MP}}.
\]
\end{proof}

\subsection[The reduced propagator is the orthogonal compression of CL]{The reduced propagator is the orthogonal compression of \texorpdfstring{$\bm{C}_L$}{CL}}
\label{app:reduced-propagator-is-compression}

Let $\bm{X}_0=\bm{U}\bm{\Sigma}\bm{V}^{\!H}$ be the SVD and let $\bm{U}_M$ be
its first $M$ columns.
Define the reduced propagator
\[
\bm{A}_M:=\bm{U}_M^{\!H}\,\bm{A}_{\mathrm{MP}}\,\bm{U}_M.
\]

\begin{lemma}
\label{lem:compression_of_CL}
With $\bm{C}_L$ from Subsection~\ref{app:existence_block_companion},
\[
\bm{A}_M=\bm{U}_M^{\!H}\bm{C}_L\bm{U}_M.
\]
\end{lemma}

\begin{proof}
From Lemma~\ref{lem:moore_penrose_projection_core},
$\bm{A}_{\mathrm{MP}}=\bm{C}_L\bm{P}_{\mathcal U}$.
Since $\mathcal U_M\subset\mathcal U$, we have $\bm P_{\mathcal U}\bm U_M=\bm U_M$, hence
\[
\bm{A}_M
=\bm{U}_M^{\!H}\bm{A}_{\mathrm{MP}}\bm{U}_M
=\bm{U}_M^{\!H}\bm{C}_L\bm{U}_M.
\]
\end{proof}

\subsection[Proof of the residual-companion identity theorem]{Proof of Theorem~\ref{thm:residual_companion_identity}}
\label{app:proof_residual_companion_identity}

\begin{proof}
Fix an eigenpair $(\widehat\lambda_j,\bm w_j)$ of $\bm A_M$ from
\eqref{eq:reduced_propagator}, and let $\widehat{\bm\phi}^{\,p}_j$ and
$\widehat{\bm\phi}^{\,e}_j$ be the associated projected and exact \ac{DMD} modes from
\eqref{eq:proj_mode_def} and \eqref{eq:exact_mode_def}.
Recall that $\widehat{\bm\phi}^{\,p}_j\in\mathcal U_M$, so
$\bm P_{\mathcal U_M}\widehat{\bm\phi}^{\,p}_j=\widehat{\bm\phi}^{\,p}_j$.

We begin by projecting the exact mode onto $\mathcal U_M$ and identifying the
reduced propagator within the expression. Using \eqref{eq:exact_mode_def},
\begin{equation}
\bm P_{\mathcal U_M}\widehat{\bm\phi}^{\,e}_j
=\bm U_M\underbrace{\bigl(\bm U_M^{\!H}\bm X_1\bm V_M\bm\Sigma_M^{-1}\bigr)}_{\bm A_M}\bm w_j \\
=\bm U_M\,\bm A_M\,\bm w_j,
\end{equation}

We now invoke the eigen-relation $\bm A_M\bm w_j=\widehat\lambda_j\bm w_j$ and the
definition of projected \ac{DMD} modes \eqref{eq:proj_mode_def} to obtain

\begin{equation}
\bm P_{\mathcal U_M}\widehat{\bm\phi}^{\,e}_j
=\bm U_M\bm A_M\bm w_j \\
=\widehat\lambda_j\,\bm U_M\bm w_j \\
=\widehat\lambda_j\,\widehat{\bm\phi}^{\,p}_j.
\end{equation}

Hence,
\begin{equation}\label{eq:residual_as_difference}
\bm r_{\mathcal U_M}\!\bigl(\widehat{\bm\phi}^{\,e}_j\bigr)
=(\bm I-\bm P_{\mathcal U_M})\widehat{\bm\phi}^{\,e}_j
=\widehat{\bm\phi}^{\,e}_j-\hat\lambda_j\,\widehat{\bm\phi}^{\,p}_j.
\end{equation}

Next, by Lemma~\ref{lem:moore_penrose_projection_core} and the existence of the
block-companion minimizer $\bm C_L$ from
Subsection~\ref{app:existence_block_companion}, we have
\[
\bm A_{\mathrm{MP}}=\bm X_1\bm X_0^\dagger=\bm C_L\,\bm P_{\mathcal U},
\qquad
\mathcal U=\operatorname{col}(\bm X_0).
\]
Since $\widehat{\bm\phi}^{\,p}_j=\bm U_M\bm w_j\in\mathcal U$, it follows that
$\bm P_{\mathcal U}\widehat{\bm\phi}^{\,p}_j=\widehat{\bm\phi}^{\,p}_j$ and thus
\[
\bm C_L\widehat{\bm\phi}^{\,p}_j
=\bm A_{\mathrm{MP}}\widehat{\bm\phi}^{\,p}_j
=\bm X_1\bm X_0^\dagger \bm U_M\bm w_j.
\]
Using the (untruncated) SVD $\bm X_0=\bm U\bm\Sigma\bm V^{\!H}$ and the fact that
$\bm U_M$ consists of the first $M$ columns of $\bm U$, we obtain
\[
\bm X_0^\dagger \bm U_M
=\bm V\bm\Sigma^{-1}\bm U^{\!H}\bm U_M
=\bm V_M\bm\Sigma_M^{-1},
\]
and therefore
\[
\bm C_L\widehat{\bm\phi}^{\,p}_j
=\bm X_1\bm V_M\bm\Sigma_M^{-1}\bm w_j
=\widehat{\bm\phi}^{\,e}_j
\]
by \eqref{eq:exact_mode_def}. Combining with \eqref{eq:residual_as_difference}
yields
\[
(\bm C_L-\hat\lambda_j\bm I)\widehat{\bm\phi}^{\,p}_j
=\widehat{\bm\phi}^{\,e}_j-\hat\lambda_j\,\widehat{\bm\phi}^{\,p}_j
=\bm r_{\mathcal U_M}\!\bigl(\widehat{\bm\phi}^{\,e}_j\bigr),
\]
which is \eqref{eq:main_identity_vector}. 
\end{proof}

%% file: appendices/fixed-eigenvalue-kv-fit.tex
\section{Fixed-Eigenvalue KV Fit (FEKVF)}
\label{app:fekvf}

The fixed-eigenvalue \ac{KV} fit (FEKVF) provides a lightweight alternative to the
nested rank-1 \ac{DMD} of Sec.~\ref{subsec:kv_nested_dmd}. After reshaping each
projected mode into a per-mode lag matrix
\[
\widehat{\bm\phi}^{\,p}_j
\;\mapsto\;
\widehat{\bm \Phi}_j
\;\coloneqq\;
\mathrm{reshape}\!\bigl(\widehat{\bm\phi}^{\,p}_j,\,D\times L\bigr)
\in\mathbb{C}^{D\times L},
\]
FEKVF evaluates how well $\widehat{\bm \Phi}_j$ conforms to a \ac{KV}
evolution that is consistent with its existing eigenvalue $\widehat{\lambda}_j$.
Unlike the nested fit, no inner DMD is performed; the eigenvalue is held fixed.

\subsection*{Closed-form residual}

For the eigenvalue $\widehat{\lambda}_j$, define the Vandermonde vector
\[
\bm a_j \;\coloneqq\; \bm v_L(\widehat{\lambda}_j)
=
[\,1,\widehat{\lambda}_j,\ldots,\widehat{\lambda}_j^{L-1}\,]^{\!\top}.
\]

Projecting $\widehat{\bm \Phi}_j$ onto the fixed-eigenvalue \ac{KV} cone amounts to the
rank-1 least-squares problem
\[
\min_{\bm u\in\mathbb C^D}\;
\bigl\|\widehat{\bm \Phi}_j-\bm u\,\bm a_j^{\!\top}\bigr\|_F^2,
\]
whose minimizer is
\[
\bm u_j^\star
\;=\;
\frac{\widehat{\bm \Phi}_j\,\bm a_j^{*}}{\|\bm a_j\|_2^2},
\qquad
\widehat{\bm \Phi}_j^{(\mathrm{KV})}
\;\coloneqq\;
\bm u_j^\star\,\bm a_j^{\!\top}.
\]
The corresponding normalized residual is
\[
\mathcal R^{(\mathrm{KV})}_j
\;\coloneqq\;
\frac{\|\widehat{\bm \Phi}_j-\widehat{\bm \Phi}_j^{(\mathrm{KV})}\|_F^2}{DL}
\]

We form a scalar \ac{KV}-conformity score
\[
\zeta_j := \log\!\big(\mathcal R^{(\mathrm{KV})}_j+\varepsilon\big),
\]
where $\varepsilon>0$ prevents singularities. Smaller values indicate closer
agreement with the KV pattern implied by $\widehat{\lambda}_j$.

\subsection*{Computation}

All operations in FEKVF reduce to one $D\times L$ matrix-vector product and a
small number of scalar reductions per mode. In particular, no inner \ac{DMD} problem
is solved and the eigenvalue is reused rather than re-estimated. As a result,
the per-mode effort is substantially smaller than for the nested rank-1 fit,
and the procedure applies independently to each mode and is easily implemented
in a batched fashion. A formal comparison of the asymptotic costs of FEKVF and
nested \ac{DMD} is given in Appendix~\ref{app:complexity}.

\begin{algorithm}
  \centering
  \setlength{\fboxsep}{6pt}
  \setlength{\fboxrule}{0.4pt}
  \fbox{%
    \begin{minipage}{0.95\linewidth}
      \begin{algorithmic}[1]
        \Require Mode matrices $\{\widehat{\bm \Phi}_j\in\mathbb{C}^{D\times L}\}_{j=1}^M$,
          eigenvalues $\{\widehat{\lambda}_j\}$, small $\varepsilon>0$
        \Ensure Per-mode scores $\{\zeta_j\}_{j=1}^M$ (smaller is more likely true)
        \For{$j=1,\ldots,M$}
          \State $\bm a \gets [\,1,\widehat{\lambda}_j,\ldots,\widehat{\lambda}_j^{L-1}\,]^{\!\top}$
          \State $\bm w \gets \bm a^{*}$
          \State $\bm u^\star \gets \widehat{\bm \Phi}_j\,\bm w / \|\bm a\|_2^2$
          \State $\widehat{\bm \Phi}^{(\mathrm{KV})}_j \gets \bm u^\star\,\bm a^{\!\top}$
          \State $\mathcal R^{(\mathrm{KV})}_j \gets \|\widehat{\bm \Phi}_j-\widehat{\bm \Phi}^{(\mathrm{KV})}_j\|_F^2/(DL)$
          \State $\zeta_j \gets \log(\mathcal R^{(\mathrm{KV})}_j+\varepsilon)$
        \EndFor
      \end{algorithmic}
    \end{minipage}}
  \caption[Fixed-Eigenvalue KV Fit (FEKVF)]{Fixed-Eigenvalue \ac{KV} Fit (FEKVF)}
  \label{alg:fekvf}
\end{algorithm}

\noindent\textbf{Remarks.}
FEKVF shares the same conceptual interpretation as the nested rank-1 \ac{DMD}: both measure global \ac{KV} conformity of the reconstructed mode trajectories.
The key difference is that FEKVF holds the eigenvalue fixed and performs a
single rank-1 projection, producing a scalar \ac{KV}-conformity score $\zeta_j$ that enters the downstream selection step shared by all mode-selection scores in this paper (Sec.~\ref{subsec:score_to_decision}).

%% file: appendices/complexity.tex
\section{Computational Complexity of the Procedures}
\label{app:complexity}

This appendix summarizes the computational costs associated with the
procedures used throughout the paper.  
Throughout, $D$ denotes the spatial dimension, $L$ the embedding length,
$DL$ the lifted dimension, and $M>m$ the truncation rank.

\subsection[DMD construction]{\ac{DMD} construction}
The delay-embedded snapshot matrices $\bm{X}_0,\bm{X}_1\in\mathbb{C}^{DL\times N}$
are processed by a rank-$M$ \ac{SVD},
\[
\bm{X}_0 \approx \bm{U}_M \bm{\Sigma}_M \bm{V}_M^\top.
\]
Computing this truncated \ac{SVD} has cost
\[
O(DL\,N\,M),
\]
which is the dominant cost in the full pipeline when $N$ is not extremely small.
The projected modes $\widehat{\bm{\phi}}^{\,p}_j$ require a single matrix–vector
multiplication per mode and cost $O(DL\,M)$ in total.  
Exact modes $\widehat{\bm{\phi}}^{\,e}_j=\bm{X}_1\bm{V}_M\bm{\Sigma}_M^{-1}e_j$
require an additional multiplication by $\bm{X}_1$, giving the same overall
$O(DL\,M)$ cost.

\subsection{Nested \ac{DMD} and block-companion construction}
The nested-\ac{DMD} procedure (Section~\ref{subsec:kv_nested_dmd}) applies a rank-1 \ac{DMD}
to each $D\times L$ mode matrix $\widehat{\bm\Phi}_j$.
The dominant operation per mode is a rank-1 \ac{SVD} of a $D\times L$ matrix, which costs $O(DL)$.
Running this for all $M$ modes, and computing each residual $\mathcal{R}^{(\mathrm{KV})}_j$
(rank-1 outer product and Frobenius-norm difference, also $O(DL)$ per mode), gives total cost
\[
O(DL\,M).
\].

\subsection{Fixed-eigenvalue \ac{KV} fit (FEKVF)}
FEKVF (Appendix~\ref{app:fekvf}) also costs $O(DL\,M)$ in total, but with a smaller
constant per mode: it replaces the rank-1 \ac{SVD} with a single matrix--vector product
$\widehat{\bm\Phi}_j\bm a_j^*\in\mathbb{C}^D$ ($O(DL)$), followed by scalar reductions
to evaluate $\mathcal{R}^{(\mathrm{KV})}_j$.
FEKVF is therefore cheaper than nested \ac{DMD} by the constant factor associated with
eliminating the inner \ac{SVD} per mode.

\subsection[ESR-score computation]{\ac{ESR}-score computation}
Each exact mode is a vector in $\mathbb{C}^{DL}$.
Using the orthogonal decomposition (Sec.~\ref{sec:mode_geometry})
\[
\widehat{\bm\phi}^{\,e}_j
\;=\;
\widehat{\lambda}_j\,\widehat{\bm\phi}^{\,p}_j
\;+\;
\bm r_{\mathcal U_M}\!\bigl(\widehat{\bm\phi}^{\,e}_j\bigr),
\qquad
\|\widehat{\bm\phi}^{\,p}_j\|_2=1,
\]
the \ac{ESR} score (residual energy) satisfies
\[
\mathcal R_j
\;=\;
\bigl\|\bm r_{\mathcal U_M}(\widehat{\bm\phi}^{\,e}_j)\bigr\|_2^2
\;=\;
\|\widehat{\bm\phi}^{\,e}_j\|_2^2 - |\widehat{\lambda}_j|^2.
\]
Thus no explicit projector applications are required, and evaluating all
$\mathcal R_j$ (and derived scalar scores such as $\zeta_j$) over $M$ modes costs
\[
O(DL\,M),
\]
with $O(M)$ additional memory.

\subsection{Score-to-decision step}
The downstream rules of Sec.~\ref{subsec:score_to_decision} operate on a per-mode scalar score $\zeta_j\in\mathbb R$. Thresholding is $O(M)$.

Clustering the $M$ scalar scores into $k=2$ groups (e.g., K-Means or a two-component GMM)
costs
\[
O(M)
\]
per iteration for K-Means, and similarly
\[
O(M)
\]
per EM iteration for a diagonal-covariance GMM. In all cases, with $k=2$, the score-to-decision cost is
negligible compared with the other steps of the pipeline.

%% file: appendices/resdmd_in_our_setting.tex
\section{ResDMD in the delay-coordinate setting}
\label{app:resdmd_in_our_setting}

\begingroup
This appendix presents the ResDMD residual~\cite{ColbrookTownsend2024}, its
instantiation to the delay-coordinate setting, and the implementation used in
the experiments. Throughout, $\bm X_0,\bm X_1\in\mathbb{C}^{DL\times N}$ are the
delay-embedded snapshot pair \eqref{eq:x0_x1}, and
$\bm U_M,\bm\Sigma_M,\bm V_M$ are the rank-$M$ truncated-\ac{SVD} matrices of
$\bm X_0$ \eqref{eq:svd_x0}.

\subsection{The ResDMD residual}
\label{app:resdmd_standard_form}

ResDMD assigns a score to a candidate eigenpair $(\lambda,g)$ of the infinite-dimensional
Koopman operator $\mathcal K$ by approximating its respective eigenresidual
$\|\mathcal K g-\lambda g\|^2_{L^2(\mu)}/\|g\|^2_{L^2(\mu)}$. Given a dictionary
$\bm\Psi=(\psi_1,\ldots,\psi_{M_\Psi})$, the candidate is
represented by its coordinate vector $\bm c\in\mathbb{C}^{M_\Psi}$ via
$g=\sum_i c_i\psi_i$. The eigenresidual is approximated by a finite-data proxy,
quadratic in $\bm c$, in which the $L^2(\mu)$ norms are replaced by quadrature
sums over the snapshots with weights $\bm W=\mathrm{diag}(w_1,\ldots,w_N)\succeq 0$.
Let $\bm\Psi_X,\bm\Psi_Y\in\mathbb{C}^{N\times M_\Psi}$ be the matrices whose rows are the dictionary
evaluated on the snapshot pairs, with Gramians
\begin{equation}
\label{eq:resdmd_gramians_standard}
\begin{aligned}
\bm G_{XX} &\coloneqq \bm\Psi_X^{\mathrm{H}}\bm W\bm\Psi_X, \\
\bm G_{YY} &\coloneqq \bm\Psi_Y^{\mathrm{H}}\bm W\bm\Psi_Y, \\
\bm G_{XY} &\coloneqq \bm\Psi_X^{\mathrm{H}}\bm W\bm\Psi_Y.
\end{aligned}
\end{equation}
The squared ResDMD residual is defined by:
\begin{equation}
\label{eq:resdmd_residual_standard}
\mathrm{res}^2(\lambda,\bm c)
\;=\;
\frac{\bm c^{\mathrm{H}}\bigl(\bm G_{YY}-\lambda\,\bm G_{XY}^{\mathrm{H}}-\lambda^{*}\,\bm G_{XY}+|\lambda|^2\bm G_{XX}\bigr)\bm c}
     {\bm c^{\mathrm{H}}\bm G_{XX}\bm c}\,,
\end{equation}
where $\lambda^{*}$ denotes the complex conjugate of $\lambda$. It was shown \cite{ColbrookTownsend2024} that ResDMD provides spectral guarantees: rigorous error control on the
Koopman pseudospectrum that avoids spectral pollution. Here it serves only as an
empirical baseline score on the true-vs-spurious ranking task of
Sec.~\ref{sec:numerical_experiments}. \textcolor{blue}{The \ac{ESR} of Sec.~\ref{subsec:resdmd_esr_companion} does not share these guarantees: it does not approximate the infinite-dimensional Koopman eigenresidual, and it does not inherit ResDMD's a posteriori, pollution-free spectral guarantees.}

\subsection{ResDMD residual in \ac{DCDMD} setting}
\label{app:resdmd_instantiated}

The instantiation of ResDMD residual in \ac{DCDMD} setting uses the linear dictionary $\bm\Psi=\bm U_M^{\mathrm H}$, corresponding to the $M$ data-driven coordinates, and sets the quadrature weights to $\bm W=\bm I$.
The dictionary matrices in \eqref{eq:resdmd_gramians_standard}
are therefore: $
\bm\Psi_X=\bm X_0^{\mathrm H}\bm U_M, \quad \bm\Psi_Y=\bm X_1^{\mathrm H}\bm U_M,
$
and the corresponding $M\times M$ Gramians are as follows:
\begin{equation}
\label{eq:resdmd_gramians_instantiated}
\begin{aligned}
\bm G_{XX} &= \bm U_M^{\mathrm H}\bm X_0\bm X_0^{\mathrm H}\bm U_M,\\
\bm G_{XY} &= \bm U_M^{\mathrm H}\bm X_0\bm X_1^{\mathrm H}\bm U_M,\\
\bm G_{YY} &= \bm U_M^{\mathrm H}\bm X_1\bm X_1^{\mathrm H}\bm U_M.
\end{aligned}
\end{equation}

The candidates $\bm c$ are obtained by diagonalizing the empirical (\ac{EDMD}~\cite{WilliamsKevrekidisRowley2015}) Koopman matrix $\bm K=\bm G_{XX}^{-1}\bm G_{XY}$
associated with the linear dictionary. Using 
\eqref{eq:resdmd_gramians_instantiated}, $\bm K$ coincides with the conjugate
transpose of the reduced \ac{DMD} propagator $\bm A_M$
\eqref{eq:reduced_propagator}:
\begin{equation}
\label{eq:resdmd_koopman_matrix}
\bm K=\bm\Sigma_M^{-1}\bm V_M^{\mathrm H}\bm X_1^{\mathrm H}\bm U_M=\bm A_M^{\mathrm H}.
\end{equation}
The coefficient vector $\bm c$ is therefore a left eigenvector of
$\bm A_M$, denoted $\bm q\in\mathbb{C}^{M}$, associated with the \ac{DMD}
eigenvalue $\widehat\lambda$, i.e. $\bm q^{\mathrm H}\bm A_M=\widehat\lambda\,\bm q^{\mathrm H}$.
\textcolor{blue}{Equivalently, $\bm q$ is an eigenvector of $\bm K=\bm A_M^{\mathrm H}$
with eigenvalue $\widehat\lambda^{*}$, since the eigenvalues of $\bm A_M^{\mathrm H}$
are the complex conjugates of those of $\bm A_M$. The candidate's Koopman eigenvalue is
thus $\widehat\lambda^{*}$, and it is this conjugate that enters the residual, while
the candidate is paired with the \ac{DMD} mode of eigenvalue $\widehat\lambda$.}

Substituting \eqref{eq:resdmd_gramians_instantiated} into \eqref{eq:resdmd_residual_standard} gives \eqref{eq:resdmd_residual_in_dcdmd}, restated here:
\begin{equation*}
\mathrm{res}^{2}(\widehat\lambda,\bm q)
=\frac{\bigl\|\bm X_1^{\mathrm H}\bm U_M \bm q-\textcolor{blue}{\widehat\lambda^{*}}\,\bm X_0^{\mathrm H}\bm U_M \bm q\bigr\|_2^{2}}
     {\bigl\|\bm X_0^{\mathrm H}\bm U_M \bm q\bigr\|_2^{2}}.
\end{equation*}

\subsection{Implementation details}
\label{app:resdmd_svd_reduction}

The score is computed as follows in our implementation
(Algorithm~\ref{alg:resdmd_baseline}; code in~\cite{harris2026_dcdmd}). Using
$\bm U_M^{\mathrm{H}}\bm X_0=\bm\Sigma_M\bm V_M^{\mathrm{H}}$, the Gramians
\eqref{eq:resdmd_gramians_standard} become
\begin{equation}
\label{eq:resdmd_specialized_gramians}
\begin{aligned}
\bm G_{XX} &= \bm\Sigma_M^{\,2}, \\
\bm G_{XY} &= \bm\Sigma_M\bm V_M^{\mathrm{H}}\bm X_1^{\mathrm{H}}\bm U_M, \\
\bm G_{YY} &= \bm U_M^{\mathrm{H}}\bm X_1\bm X_1^{\mathrm{H}}\bm U_M,
\end{aligned}
\end{equation}
and the squared residual in \eqref{eq:resdmd_residual_in_dcdmd} is the Gramian-quadratic
form
\begin{equation}
\label{eq:resdmd_gramian_form}
\mathrm{res}^{2}(\widehat\lambda,\bm q)
\;=\;
\frac{\bm q^{\mathrm{H}}\,\bm R(\textcolor{blue}{\widehat\lambda^{*}})\,\bm q}{\bm q^{\mathrm{H}}\bm G_{XX}\bm q}\,,
\end{equation}
where 
\begin{equation} 
\bm R(\lambda)\;\coloneqq\;\bm G_{YY}-\lambda\,\bm G_{XY}^{\mathrm{H}}-\lambda^{*}\,\bm G_{XY}+|\lambda|^2\bm G_{XX}\,.
\end{equation}

\begin{algorithm}
  \centering
  \setlength{\fboxsep}{6pt}
  \setlength{\fboxrule}{0.4pt}
  \fbox{%
    \begin{minipage}{0.95\linewidth}
      \begin{algorithmic}[1]
        \Require Delay snapshots $\bm X_0,\bm X_1$; truncation rank $M$; small $\varepsilon>0$
        \Ensure Per-mode scores $\{s_j\}_{j=1}^M$ (smaller is more likely true)
        \State $(\bm U_M,\bm\Sigma_M,\bm V_M) \gets$ rank-$M$ \ac{SVD} of $\bm X_0$
        \State $\bm A_M \gets \bm U_M^{\mathrm{H}}\bm X_1\bm V_M\bm\Sigma_M^{-1}$
        \State $\{(\widehat\lambda_j,\bm q_j)\}_{j=1}^M \gets$ \textcolor{blue}{left eigenpairs of $\bm A_M$ (equivalently, eigenvectors of $\bm A_M^{\mathrm{H}}$ with eigenvalues $\widehat\lambda_j^{*}$)}
        \State form $\bm G_{XX},\bm G_{XY},\bm G_{YY}$ as in \eqref{eq:resdmd_specialized_gramians}
        \For{$j=1,\ldots,M$}
          \State $\mathcal R_j \gets \dfrac{\bm q_j^{\mathrm{H}}\,\bm R(\textcolor{blue}{\widehat\lambda_j^{*}})\,\bm q_j}{\bm q_j^{\mathrm{H}}\bm G_{XX}\bm q_j}$
          \State $s_j \gets \log(\mathcal R_j+\varepsilon)$
        \EndFor
      \end{algorithmic}
    \end{minipage}}
  \caption{ResDMD baseline score on \ac{DCDMD} eigenpairs, with the Hermitian
  residual matrix $\bm R(\lambda)$ as in \eqref{eq:resdmd_gramian_form}.}
  \label{alg:resdmd_baseline}
\end{algorithm}

Each candidate shares its eigenvalue \textcolor{blue}{(up to complex conjugation)} with a \ac{DMD} mode, so its score $s_j$
is assigned to that mode; conjugate pairs are scored separately.
A candidate's residual can be measured either on the same snapshots used to
build the operator $\bm A_M$, or on a separate, held-out batch of snapshots
independent of that fit (which keeps the residual from being optimistically
small). The instantiation of ResDMD in \ac{DCDMD} setting follows the former, where the same pair $(\bm X_0,\bm X_1)$ is used throughout.

\subsection{Projection form}
\label{app:resdmd_projection_form}

To compare the ESR and ResDMD, Sec.~\ref{subsec:resdmd_esr_companion} expresses the squared norm of the ResDMD residual numerator as a projection. The derivation of this projection form is as follows. The SVD truncation yields
\[
\bm X_0^{\mathrm{H}}\bm U_M
=\bm V\bm\Sigma\bm U^{\mathrm{H}}\bm U_M
=\bm V_M\bm\Sigma_M ,
\]
and the eigenrelation
\[
\textcolor{blue}{\widehat\lambda^{*}}\,\bm q=\bm A_M^{\mathrm{H}}\bm q
=\bm\Sigma_M^{-1}\bm V_M^{\mathrm{H}}\bm X_1^{\mathrm{H}}\bm U_M\bm q
\]
can be rearranged as
$\bm\Sigma_M\textcolor{blue}{\widehat\lambda^{*}}\,\bm q=\bm V_M^{\mathrm{H}}\bm X_1^{\mathrm{H}}\bm U_M\bm q$.
The numerator is therefore
\begin{equation}
\label{eq:resdmd_projection_form}
\begin{aligned}
\bm X_1^{\mathrm H}\bm U_M\bm q
-
\textcolor{blue}{\widehat\lambda^{*}}\,\bm X_0^{\mathrm H}\bm U_M\bm q
&=
\bm X_1^{\mathrm H}\bm U_M\bm q
-
\textcolor{blue}{\widehat\lambda^{*}}\,\bm V_M\bm\Sigma_M\bm q
\\
&=
\bm X_1^{\mathrm H}\bm U_M\bm q
-
\bm V_M\bm V_M^{\mathrm H}\bm X_1^{\mathrm H}\bm U_M\bm q
\\
&=
\bigl(\bm I-\bm V_M\bm V_M^{\mathrm H}\bigr)
\bm X_1^{\mathrm H}\bm U_M\bm q .
\end{aligned}
\end{equation}

\endgroup